\documentclass[11pt]{article}
\usepackage{amsfonts,amsthm,amssymb}
\usepackage{cite}
\usepackage[pdftex]{graphicx}
\usepackage{color}
\usepackage[fleqn]{amsmath}
\usepackage{hyperref}
\newcommand{\be}{\begin{eqnarray}}
\newcommand{\ee}{\end{eqnarray}}
\newcommand{\bez}{\begin{eqnarray*}}
\newcommand{\eez}{\end{eqnarray*}}
\renewcommand{\d}{\mathrm{d}}
\newcommand{\bd}{\bar{\mathrm{d}}}

\newcommand{\bsy}{\boldsymbol}
\newcommand{\imag}{\mathrm{i} \,}
\theoremstyle{plain}
\newtheorem{theorem}{Theorem}[section]
\newtheorem{lemma}[theorem]{Lemma}
\newtheorem{proposition}[theorem]{Proposition}
\newtheorem{corollary}[theorem]{Corollary}

\theoremstyle{definition}

\newtheorem{remark}[theorem]{Remark}
\newtheorem{example}[theorem]{Example}

\numberwithin{equation}{section}
\numberwithin{theorem}{section}

\allowdisplaybreaks

\begin{document}

\title{\textbf{Miura transformation in bidifferential calculus \\
and a vectorial Darboux transformation for the 
Fokas-Lenells equation}} 
\author{
 {\sc Folkert M\"uller-Hoissen}$\,^a$ and {\sc Rusuo Ye}$\,^b$ \\ 
 \small 
 $^a$ Institut für Theoretische Physik, Friedrich-Hund-Platz 1,
 37077 G\"ottingen, Germany \\
 \small E-mail: folkert.mueller-hoissen@theorie.physik.uni-goettingen.de \\
 \small $^b$ College of Mathematics, Wenzhou University, Wenzhou	325035, PR China \\
 \small E-mail: rusuoye@163.com
}

\date{} 

\maketitle

\begin{abstract}
Using a general result of bidifferential calculus and recent results of other authors, a vectorial binary Darboux transformation is derived for the first member of the ``negative" part of the potential Kaup-Newell hierarchy, which is a system of two coupled Fokas-Lenells equations. Miura transformations are found 
from the latter to the first member of the negative part of the AKNS hierarchy 
and also to its ``pseudodual". 
The reduction to the Fokas-Lenells equation is implemented and exact solutions 
with a plane wave seed generated.
\end{abstract}    

\section{Introduction}
In a (by transformations of variables) simplified form \cite{Lene09}, the \emph{Fokas-Lenells equation} \cite{Fokas95,Lene+Foka09} reads\footnote{The sign 
of the last term changes via complex conjugation, $u \mapsto u^\ast$, and also via $(x,t) \mapsto (-x,-t)$. Hence 
(\ref{FLeq}) is CPT-invariant. Whereas in
the case of its relative, the nonlinear Schr\"odinger (NLS) equation, one has to distinguish between a ``focusing" and a ``defocusing" case, with qualitatively very 
different solutions, this is not the case here \cite{Lene09}. Via the formula in 
Proposition~1 in \cite{Lene09} with $\sigma=-1$, each solution of (\ref{FLeq}) determines a solution 
of the original Fokas-Lenells equation (2) in \cite{Lene09}. } 
\be
     u_{xt} - u - 2 \imag |u|^2 u_x = 0  \, , \label{FLeq}
\ee
where $u$ is a complex function of independent real variables $x$ and $t$, and a
subscript $x$ or $t$ indicates a partial derivative with respect to $x$,  respectively $t$. It is a model for the propagation of nonlinear
light pulses in monomode optical fibers, taking certain nonlinear effects into account. 
\vspace{.2cm}

The Fokas-Lenells equation belongs to the class of ``completely integrable" nonlinear partial differential equations (PDEs), which arise as integrability condition of a system of two linear equations (``Lax pair"). It is the first 
equation of the `negative' part of the derivative nonlinear Schr\"odinger (DNLS) 
hierarchy and arises as a reduction of the 
first member of the ``negative" part of the potential Kaup-Newell hierarchy
\cite{GIK80,Yang+Zeng07,Guo+Ling12,FGZ13}, 
\be
 u_{xt} - u - 2 \imag u v u_x = 0 \, , \qquad
 v_{xt} - v + 2 \imag u v v_x = 0 \, ,   \label{pKN-1}
\ee
via the constraint $v = u^\ast$. It is more convenient, however, to refer to this system as \emph{coupled Fokas-Lenells equations}.\footnote{However, it should be noticed that, in the literature, the names ``Fokas-Lenells system" or ``coupled Fokas-Lenells equations" also refer to different systems, related to original forms of the Fokas-Lenells equation \cite{ZYCW17,LFZ18,KXM18,CYSGB18,YZCBG19,Ling+Su24}, whereas (\ref{pKN-1}) is sometimes referred to as the ``pKN(-1) system".} 
\vspace{.2cm}

Each solution of the Fokas-Lenells equation determines a solution of the (also integrable) two-dimensional massive Thirring model (see Appendix~A in \cite{LWZ22} and references cited there). 
\vspace{.2cm}

Exact solutions of (\ref{FLeq}) have been 
produced by various methods, like Hirota's bilinearization \cite{Veks11,Mats12a,Mats12b,LWZ22,LZLX20,DTSN23}, inverse scattering 
\cite{Zhao+Fan21,Cheng+Fan23}, Riemann-Hilbert problem \cite{Guo+Ling12,Ai+Xu19,Zhang+Tian23,ZQH23}, algebro-geometric 
techniques \cite{ZFH13} (solutions in terms of the Riemann theta function), \cite{Zhao+Wang22} (solutions in terms of Jacobi elliptic functions), dressing and Darboux transformations \cite{Lene10,HXP12,XHCP15,ZYCW17,WXL20}.
The latter constitute a powerful method to generate large classes of exact solutions of an integrable PDE (or, more generally, a partial differential-difference equation) from given solutions \cite{Matv+Sall91}. 
A substantial improvement is a \emph{vectorial} version \cite{Manas96}. 
Multiple soliton solutions are obtained with it in a single step, using 
diagonal matrix data, whereas repeated application is needed in case of a scalar Darboux transformation. Furthermore, larger classes of solutions are reached, since the method can be applied with non-diagonal matrix data.   
\vspace{.2cm}

For the class of integrable equations possessing a ``bidifferential calculus representation" \cite{DMH00a}, there is a universal vectorial Darboux transformation. 
Given an integrable equation from this class, it is usually straight forward 
to deduce the corresponding vectorial Darboux transformation from the universal 
one.
Such an approach was taken in \cite{Ye+Zhang23}, where, however, the 
coupled Fokas-Lenells equations were only reached via a Miura transformation.\footnote{In the present context, a ``Miura transformation" is 
understood as a relation expressing the dependent variable of a differential equation in terms of the dependent variable of another differential equation, and its derivatives, implying that any solution of the latter determines a solution of the former equation. It is named after Robert M. Miura, who discovered such 
a relation between the KdV and the modified KdV equation.}	
As a consequence, the computations turned out to be overly complicated 
and the results obtained were not sufficiently general. By using observations in \cite{LLZ24}, 
more powerful results can be achieved in a much more elegant way, which also 
greatly improve those in \cite{LLZ24}, where the so-called Cauchy matrix 
approach has been taken to generate exact solutions. 
\vspace{.2cm}

Section~\ref{sec:bidiff} briefly recalls some basics of bidifferential 
calculus and a binary Darboux transformation \cite{DMH13SIGMA} in this framework. 
As suggested by the results in \cite{LLZ24}, we consider transformations of the Miura transformation equation, in bidifferential calculus, instead of the integrable equations which it connects. 
The crucial point is that the binary Darboux transformation also generates 
new solutions of the Miura transformation equation from a given 
solution.\footnote{This fact has been known to the first author and his late colleague Aristophanes Dimakis since the work in \cite{DMH13SIGMA}. But it 
has not been exploited so far.} Theorem~\ref{thm:Miura} is our first 
main result and offers applications far beyond the one presented in this work.
\vspace{.2cm}

A straight elaboration of this general result for the case of a special bidifferential calculus, in Section~\ref{sec:FLsys}, leads to our second main result, a 
vectorial binary Darboux transformation for the coupled Fokas-Lenells equations
(Theorem~\ref{thm:bDT_KN-1}). 
In Section~\ref{subsec:Miura} we show that, indeed, concrete Miura transformations are obtained from elaboration of the Miura transformation in bidifferential calculus.
 \vspace{.2cm}
 
In Section~\ref{sec:red}, the reduction to the Fokas-Lenells equation is implemented in the binary Darboux transformation. The resulting 
vectorial Darboux transformation for the Fokas-Lenells equation is our third main result (Theorem~\ref{thm:FL_DT}). We apply it in the case of vanishing seed, and in Section~\ref{sec:FLeq_pw_seed} also to the case of a plane wave seed. In this way, we recover all known classes of solitons of the Fokas-Lenells equation in a straight way.
\vspace{.2cm}

Section~\ref{sec:concl} contains some final remarks.

\section{Binary Darboux transformations in bidifferential calculus}
\label{sec:bidiff}
A \emph{graded associative algebra} is an associative algebra 
$\boldsymbol{\Omega} = \bigoplus_{r \geq 0} \boldsymbol{\Omega}^r$
over a field $\mathbb{K}$ of characteristic zero, where $\mathcal{A} := \boldsymbol{\Omega}^0$ is an associative 
algebra over $\mathbb{K}$ and $\boldsymbol{\Omega}^r$, $r \geq 1$, are $\mathcal{A}$-bimodules such that 
$\boldsymbol{\Omega}^r \, \boldsymbol{\Omega}^s \subseteq \boldsymbol{\Omega}^{r+s}$. 
Elements of $\boldsymbol{\Omega}^r$ are called $r$-forms.
A \emph{bidifferential calculus} is a unital graded associative algebra $\boldsymbol{\Omega}$, supplied
with two $\mathbb{K}$-linear graded derivations 
$\d, \bd : \boldsymbol{\Omega} \rightarrow \boldsymbol{\Omega}$
of degree one (so that $\d \boldsymbol{\Omega}^r \subseteq \boldsymbol{\Omega}^{r+1}$,
$\bd \boldsymbol{\Omega}^r \subseteq \boldsymbol{\Omega}^{r+1}$), and such that
\bez
    \d^2 = 0 \, , \qquad  \bd^2 = 0 \, , \qquad  \d \bd + \bd \d = 0  \, .   
\eez
We refer the reader to \cite{DMH20dc} for an introduction to this structure and an extensive list of references. 

The following results (also see \cite{DMH13SIGMA,CDMH16,DMH20dc}) 
express the essence of binary Darboux 
transformations that have been found for many integrable 
partial differential and difference equations. A crucial 
advantage is that, on the 
level of bidifferential calculus, proofs 
are very simple as a consequence of the simple computational 
rules and properties of $\d$ and $\bd$.

\begin{theorem}
\label{thm:A}
Given a bidifferential calculus with maps $\d, \bd$, 
let 0-forms $\Delta, \Gamma$ and 1-forms  $\kappa,\lambda$ satisfy  
\be
  &&  \bd \Delta + [\lambda , \Delta] = (\d \Delta) \, \Delta \, , \qquad  
      \bd \lambda + \lambda^2 = (\d \lambda) \, \Delta \, , \nonumber \\ 
  &&  \bd \Gamma - [\kappa , \Gamma] = \Gamma \, \d \Gamma  \, , \qquad  
      \bd \kappa - \kappa^2 = \Gamma \, \d \kappa \, .   \label{Delta,lambda,Gamma,kappa_eqs}
\ee
Let 0-forms $\theta$ and $\eta$ be solutions of the linear equations 
\be
    \bd \theta = A \, \theta + (\d \theta) \, \Delta + \theta \, \lambda \, , \qquad
   \bd \eta = - \eta \, A + \Gamma \, \d \eta + \kappa \, \eta \, ,   \label{linsys1}
\ee
where the 1-form $A$ satisfies
  \be
         \d A = 0 \, , \qquad \bd A = A^2 \, ,    \label{A_eqs}
\ee
and $\Omega$ an invertible solution of the linear system
\be
    && \Gamma \, \Omega - \Omega \, \Delta = \eta \, \theta \, ,  \label{preSylv} \\
    && \bd \Omega = (\d \Omega) \, \Delta - (\d \Gamma) \, \Omega + \kappa \, \Omega + \Omega \, \lambda + (\d \eta) \, \theta \, .
        \label{linsys2}
\ee
Then, if $\Omega$ is invertible, 
\be
        A' := A - \d (\theta \, \Omega^{-1} \eta)      \label{A'}
\ee
also solves (\ref{A_eqs}). \hfill $\Box$
\end{theorem}

\begin{corollary}
\label{cor:phi}
Let (\ref{Delta,lambda,Gamma,kappa_eqs}) hold and (\ref{linsys1}) with $A=\d \phi$, where the 0-form $\phi$ 
is a solution of
\be
         \bd \d \phi = \d \phi \, \d \phi \, .   \label{phi_eq}
\ee
If $\Omega$ is an invertible solution of (\ref{preSylv}) and (\ref{linsys2}), then  
\be
        \phi' = \phi - \theta \, \Omega^{-1} \eta + C \, ,
         \label{phi'}
\ee
where $C$ is any $\d$-constant (i.e., $\d C=0$), solves the same equation.
\hfill $\Box$
\end{corollary}

\begin{corollary}
\label{cor:g}
Let (\ref{Delta,lambda,Gamma,kappa_eqs}) hold with invertible $\Delta$ and $\Gamma$, and (\ref{linsys1}) 
with $A= (\bd g) \, g^{-1}$, where $g$ is a solution of
\be
         \d ( (\bd g) \, g^{-1}) = 0 \, .       \label{g_eq}
\ee
Let $\Omega$ be an invertible solution of (\ref{preSylv}) and (\ref{linsys2}).  Then
\be
        g' = g - \theta \,  \Omega^{-1} \Gamma^{-1} \eta \, g 
         \label{g'}
\ee
solves the same equation. \hfill $\Box$
\end{corollary}

The results in Corollary \ref{cor:phi} and Corollary \ref{cor:g} can be regarded as reductions of that in Theorem \ref{thm:A}. 
The above results remain true if the objects are matrices of forms of appropriate 
sizes, so that all appearing products and also the actions of $\d$ and $\bd$ are defined. 

An additional result that has not yet been paid attention to, is the following. 

\begin{theorem}
	\label{thm:Miura}
Let the conditions in Theorem~\ref{thm:A} hold. Let $\phi$ and an invertible $g$ satisfy the Miura tranformation equation
\be
      (\bd g) \, g^{-1} = \d \phi       \label{Miura}
\ee
(which connects the two equations (\ref{phi_eq}) and (\ref{g_eq}), and has both as 
integrability condition).
Then also $(\phi',g')$, given by (\ref{phi'}) and (\ref{g'}), solve this equation. 
\end{theorem}
\begin{proof}
Using the derivation rule for $\bd$ on 0-forms and (\ref{Delta,lambda,Gamma,kappa_eqs}), (\ref{linsys1}), we obtain
\bez
   \bd g' &=& \bd g - \theta \Omega^{-1} \Gamma^{-1} \eta \, ( \bd g - \d \phi \, g )
    -  \theta \Omega^{-1} \d \eta \, g + \theta \Omega^{-1} (  \d \Omega \, \Delta \Omega^{-1} + \theta \, \Omega^{-1} \d \eta \, \theta )  \,   
       \Omega^{-1} \Gamma^{-1} \eta \, g \\
  &&  - ( \d \phi \, \theta + \d \theta \, \Delta ) \, \Omega^{-1} \Gamma^{-1} \eta \, g  \, , 
\eez
where the bracket in the second term on the right hand side vanishes by assumption. Then
\bez
     \bd g' - (\d \phi') \, g' &=& \Big( \theta \Omega^{-1} \d \Omega \, \Delta - \d \theta \, \Delta 
          - \d \theta \, \Omega^{-1} \eta \theta + \theta \Omega^{-1} \d \Omega \, \Omega^{-1} \eta \theta 
                 \Big) \, \Omega^{-1} \Gamma^{-1} \eta \, g  \\
       && + ( \d \theta \, \Omega^{-1}  - \theta \, \Omega^{-1} \d \Omega \, \Omega^{-1} ) \, \eta \, g 
\eez
vanishes by substituting for $\eta \, \theta$ the left hand side of (\ref{preSylv}).
\end{proof}

Under the conditions of Theorem~\ref{thm:Miura}, Corollary~\ref{cor:phi} and 
Corollary~\ref{cor:g} are direct consequences of the last theorem. It should be mentioned, however, that 
only for special bidifferential calculi (\ref{Miura}) leads to a meaningful equation, whereas (\ref{phi_eq}) and (\ref{g_eq}) have a better standing in this respect. 
If, for some bidifferential calculus, (\ref{Miura}) results in an equation possessing nontrivial solutions, 
the relation between (\ref{phi_eq}) and (\ref{g_eq}), expressed by it, is sometimes called ``pseudoduality".  

A relevant application of Theorem~\ref{thm:Miura} is provided in Section~\ref{sec:FLsys}, based on \cite{LLZ24}.

\section{Application to the coupled Fokas-Lenells equations}
\label{sec:FLsys}
Let $\mathcal{A}$ be the associative algebra of smooth functions of independent real variables $x$ and $t$. 
We choose $\bsy{\Omega} = \mathrm{Mat}(\mathcal{A}) \otimes \bigwedge(\mathbb{C}^2)$,
where $\mathrm{Mat}(\mathcal{A})$ is the algebra of all matrices\footnote{The product of two matrices is set to zero if the sizes do not match.} 
over $\mathcal{A}$, and a basis $\xi_1,\xi_2$ of the Grassmann algebra $\bigwedge(\mathbb{C}^2)$. It is sufficient to define $\d$ and $\bd$ on $\mathrm{Mat}(\mathcal{A})$. Then they extend in an obvious way to $\Omega$ (treating $\xi_1,\xi_2$ as constants).

Let $\phi$ and $g$ be $2 \times 2$ matrices over $\mathcal{A}$. For each $n>1$, let $J_n$ be a constant $n \times n$ matrix such that $J_n^2 = I_n$ (the $n \times n$ identity matrix) and $J_n \neq I_n$. 
For an $m \times n$ matrix $F$ over $\mathcal{A}$, we set
\bez
\d F = F_x \, \xi_1 + \frac{1}{2} (J_m F - F J_n) \, \xi_2 \, , \qquad
\bd F = \frac{1}{2} (J_m F - F J_n) \, \xi_1 + F_t \, \xi_2 
\eez
(also see \cite{DMH10AKNS,DKMH11acta,CDMH16,MH23,Ye+Zhang23}). We write $J = J_2$.
Then the Miura transformation equation (\ref{Miura}) becomes the 
``Miura system"\footnote{A similar system has been explored in \cite{GKM24}.}
\be
  \phi_x = \frac{1}{2} \, [J,g] \, g^{-1} \, , \qquad g_t g^{-1} = \frac{1}{2} \, [J,\phi] 
  \, .   \label{Miura_sys}
\ee
Introducing 
\bez
     \tilde{g} := \frac{1}{\sqrt{\det(g)}} \, g \, ,
\eez
the latter equations become
\bez
  \phi_x = \frac{1}{2} \, [J,\tilde{g}] \, \tilde{g}^{-1} \, , \qquad
  \tilde{g}_t \tilde{g}^{-1} + \frac{(\sqrt{\det(g)})_t}{\sqrt{\det(g)}} \, I_2 = \frac{1}{2} \, [J,\phi] \, .
\eez
From the second equation, by taking the trace, we obtain
\bez
        \det(g)_t = 0 \, .
\eez
Let now $J = \mathrm{diag}(1,-1)$. Then the Miura system (\ref{Miura_sys}) results in 
\be
  \left( \begin{array}{cc} \tilde{g}_{11} & \tilde{g}_{12} \\
    	\tilde{g}_{21} & \tilde{g}_{22} \end{array} \right)_t 
  = \left( \begin{array}{cc} \phi_{12} \, \tilde{g}_{21} & \phi_{12} \,  \tilde{g}_{22} \\
  	- \phi_{21} \, \tilde{g}_{11} & - \phi_{21} \, \tilde{g}_{12} \end{array} \right)\, , \quad
  \left( \begin{array}{cc} \phi_{11} & \phi_{12} \\
  	\phi_{21} & \phi_{22} \end{array} \right)_x 
 = \left( \begin{array}{cc} - \tilde{g}_{12} \, \tilde{g}_{21} & \tilde{g}_{12} \, \tilde{g}_{11} \\
 	- \tilde{g}_{21} \tilde{g}_{22} & \tilde{g}_{12} \, \tilde{g}_{21} \end{array} \right)\, .	\label{Miura_sys2} 	
\ee
These are, up to renamings, the two equations (2.15) in \cite{LLZ24}. 

\begin{proposition}[\cite{LLZ24}]
\label{prop:pKN-1} 
Let $\phi$ and $\tilde{g}$, with $\tilde{g}_{22} \neq 0$, satisfy the Miura system. Then
\be
    u := \phi_{21} \, , \qquad 
    v := \imag \tilde{g}_{12}/\tilde{g}_{22} 
       = \imag g_{12}/g_{22}   \label{u,v} \, , 
\ee	
solve the coupled Fokas-Lenells equations (\ref{pKN-1}). 
\end{proposition}
\begin{proof}
This is easily verified directly, using (\ref{Miura_sys2}) and 
$\det(\tilde{g}) = 1$.
\end{proof}

\begin{remark}
Whereas the authors of \cite{LLZ24} started by reducing the Miura transformation, relating the two familiar potential 
forms of the self-dual Yang-Mills (sdYM) equation, to a two-dimensional system, more directly a corresponding reduction of a bidifferential calculus 
for the sdYM equation has been used in \cite{Ye+Zhang23}. 
The vectorial Darboux transformation for the coupled Fokas-Lenells equations, obtained in \cite{Ye+Zhang23}, will be replaced in the present work by a much better version, using results from \cite{LLZ24}. The crucial new insight in \cite{LLZ24} is in fact what is formulated in the preceding proposition. 
\hfill $\Box$
\end{remark}

\begin{lemma}
	\label{lem:Miura_uv} \hspace{1cm}  \\
(1) If $u$, $u_x$ and $v$ are non-zero, the Miura system (\ref{Miura_sys2}) implies
\be
  && \tilde{g}_{22} = \exp\Big( \imag \int u v \, dt \Big) \, , \label{Miura_1} \\
  && \tilde{g}_{11} = \frac{1+ \imag u_x v}{\tilde{g}_{22}} \, , \qquad
     \tilde{g}_{12} = - \imag v \, \tilde{g}_{22} \, , \qquad
     \tilde{g}_{21} = - \frac{u_x}{\tilde{g}_{22}} \, , \label{Miura_2} \\
  && \phi_{12} = - \imag v_t + u v^2 \, , \qquad \phi_{21} = u \, , \label{Miura_3} \\
  && \phi_{11x} = - \phi_{22x} = - \imag u_x v \, . \label{Miura_4}
\ee
(2) If $(u,v)$ solves the coupled Fokas-Lenells equations (\ref{pKN-1}), then (\ref{Miura_1})-(\ref{Miura_4}) 
determines a solution $(\phi,\tilde{g})$ of the Miura system (\ref{Miura_sys2}).
\end{lemma}
\begin{proof}
Eliminating $\phi_{21}$ and $\tilde{g}_{12}$ via (\ref{u,v}), (\ref{Miura_sys2}) is equivalent to 
\bez
&& \tilde{g}_{22t} = \imag u \, v \, \tilde{g}_{22} \, , \qquad
\tilde{g}_{21} \, \tilde{g}_{22} = - u_x  \, , \qquad 
u \, \tilde{g}_{11} = - \tilde{g}_{21t} \, , \qquad 
 \phi_{12} = - \imag v_t + u \, v^2 \, , \nonumber \\
&& \phi_{11x} = - \phi_{22x} = -  \tilde{g}_{12} \, \tilde{g}_{21} 
= - \imag u_x v \, , \qquad
\phi_{12x} = \tilde{g}_{12} \, \tilde{g}_{11} \, , \qquad
\phi_{12} \, \tilde{g}_{21} = \tilde{g}_{11t} \, .  
\eez	
The general solution of the first equation is given by (\ref{Miura_1}) and everywhere non-zero on $\mathbb{R}^2$. \\
(1) Let  (\ref{Miura_sys2}) hold. According to Proposition~\ref{prop:pKN-1}, $u$ and $v$ then satisfy (\ref{pKN-1}). The last system determines, in turn, $\tilde{g}_{12}$, $\tilde{g}_{21}$, $\tilde{g}_{11}$ (assuming 
$u \neq 0$ and using (\ref{pKN-1}))  in terms of $u$ and $v$. The obtained expressions are in (\ref{Miura_2}). 
The last two equations of the above system are then satisfied as a consequence of (\ref{pKN-1}), 
and the remaining equations are given by (\ref{Miura_4}).  \\
(2)  This is easily verified.
\end{proof}	
	
Next we present a vectorial binary Darboux transformation for the coupled Fokas-Lenells equations (\ref{pKN-1}).

\begin{theorem}
\label{thm:bDT_KN-1}
Let $\Delta$ and $\Gamma$ be invertible constant $n \times n$ matrices. Let $(u,v)$ be a solution of the coupled Fokas-Lenells equations (\ref{pKN-1}). 
Furthermore, let $\theta_1,\theta_2$ and $\eta_1,\eta_2$ be $n$-component row, respectively column vector solutions of the linear systems
\be
&&   \theta_{1x} \, \Delta - \theta_1 \, (\frac{1}{2} + \imag u_x v) 
     - \imag \theta_2 \, v \, ( 1 + \imag u_x v) = 0 \, , \qquad
   \theta_{2x} \, \Delta + \theta_2 \, (\frac{1}{2} + \imag u_x v) + \theta_1 \, u_x = 0 \, , \nonumber \\
&&  \theta_{1t} - \frac{1}{2} \theta_1 \Delta + \theta_2 \, (\imag v_t - u v^2) = 0 \, , \qquad
    \theta_{2t} + \frac{1}{2} \theta_2 \Delta + \theta_1 \, u = 0 \, , \nonumber \\
&&  \Gamma \, \eta_{1x} + (\frac{1}{2} + \imag u_x v) \, \eta_1 - u_x \, \eta_2 = 0 \, , \qquad
	\Gamma \, \eta_{2x} - (\frac{1}{2} + \imag u_x v) \, \eta_2 
	+ \imag v \, (1 + \imag u_x v) \, \eta_1 = 0 \, ,   \nonumber \\
&& \eta_{1t} + \frac{1}{2} \Gamma \, \eta_1 - u \, \eta_2 = 0 \, , \qquad
   \eta_{2t} - \frac{1}{2} \Gamma \, \eta_2 - (\imag v_t - u v^2) \, \eta_1 = 0 \, .
     \label{FL_linsys}
\ee
Furthermore, let $\Omega$ be an $n \times n$ matrix solution of the Sylvester equation 
\be
    \Gamma \, \Omega - \Omega \, \Delta = \eta_1 \, \theta_1 + \eta_2 \, \theta_2 
    \, ,  \label{Sylv}
\ee
and the linear equations
\be
 \Omega_x \Delta = - \eta_{1x} \theta_1 - \eta_{2x} \theta_2 \, , \qquad   
 \Omega_t = - \frac{1}{2} ( \eta_1 \, \theta_1 - \eta_2 \, \theta_2 ) \, .   \label{FL_Om}
\ee
Then, in any open set of $\mathbb{R}^2$ where $\Omega$ is invertible,
\be
   u' := u - \theta_2 \Omega^{-1} \eta_1 \, , 
   \qquad 
   v' := \frac{ v \, (1 - \theta_1 \Omega^{-1} \Gamma^{-1} \eta_1 )  
   	 -\imag \theta_1 \Omega^{-1} \Gamma^{-1} \eta_2 }{ 1 - \theta_2 \Omega^{-1} \Gamma^{-1} \eta_2 + \imag v \, \theta_2 \Omega^{-1} \Gamma^{-1} \eta_1 } \, ,
  \label{u',v'}
\ee
solve the coupled Fokas-Lenells equations (\ref{pKN-1}).
\end{theorem}
\begin{proof}
Writing
\bez
A = A_1 \xi_1 + A_2 \xi_2 \, , \qquad
\kappa = \kappa_1 \xi_1 + \kappa_2 \xi_2 \, , \qquad
\lambda = \lambda_1 \xi_1 + \lambda_2 \xi_2 \, ,
\eez 
(\ref{linsys1}) reads
\bez
&&   \frac{1}{2} (J_2 \theta - \theta J_n) = A_1 \theta + \theta_x \Delta + \theta \lambda_1 \, , \qquad
\theta_t = A_2 \theta + \frac{1}{2} (J_2 \theta - \theta J_n) \Delta + \theta \lambda_2 \, , \\
&&   \frac{1}{2} (J_n \eta - \eta J_2) = - \eta A_1 + \Gamma \eta_x + \kappa_1 \eta  \, , \qquad
\eta_t = - \eta A_2 + \frac{1}{2} \Gamma (J_n \eta - \eta J_2) + \kappa_2 \eta  \, . 
\eez
Choosing 
\bez
\kappa_1 = \frac{1}{2} J_n \, , \qquad
\kappa_2 = - \frac{1}{2} \Gamma J_n \, , \qquad
\lambda_1 = - \frac{1}{2} J_n \, , \qquad
\lambda_2 = \frac{1}{2} J_n \Delta \, , 
\eez
and using
\bez
      A_1 = \phi_x \, , \qquad
      A_2 = \frac{1}{2} [J , \phi] = \left( \begin{array}{cc} 0 & \phi_{12} \\ - \phi_{21} & 0 \end{array} \right) \, ,
\eez
the latter system simplifies to 
\bez
&&   \frac{1}{2} J \, \theta = \phi_x \, \theta + \theta_x \Delta  \, , \qquad 
\theta_t = \frac{1}{2} [J , \phi] \, \theta + \frac{1}{2} J \, \theta \Delta  \, ,  \\
&&   \frac{1}{2} \eta \, J =  \eta \, \phi_x - \Gamma \eta_x \, , \qquad
\eta_t = - \frac{1}{2} \eta \, [J , \phi] - \frac{1}{2} \Gamma \eta \, J \, .
\eez
Writing 
\bez
    \theta = \left( \begin{array}{c} \theta_1 \\ \theta_2 \end{array} \right) \, , 
    \qquad
    \eta = \left( \begin{array}{cc} \eta_1 & \eta_2 \end{array} \right) \, ,
\eez
with $n$-component row vectors $\theta_i$ and $n$-component column vectors $\eta_i$,
$i=1,2$, and using the solution of the Miura system corresponding to $(u,v)$  
according to Lemma~\ref{lem:Miura_uv}, this becomes (\ref{FL_linsys}). 
We note that $\phi_{12x} = -\imag v \, ( 1 + \imag u_x v )$.

The conditions in (\ref{Delta,lambda,Gamma,kappa_eqs}) boil down to
\bez
   \Delta_x = \Delta_t = 0 \, , \qquad \Gamma_x = \Gamma_t = 0 \, ,
\eez
so that the $n \times n$ matrices $\Delta$ and $\Gamma$ have to be constant.
(\ref{linsys2}) takes the form 
\bez
   \Omega_x \Delta = - \eta_x \theta \, , \qquad   
   \Omega_t = - \frac{1}{2} \eta \, J \, \theta \, ,
\eez
which is (\ref{FL_Om}). To derive the second equation, we used the Sylvester equation. The statement of the theorem now follows from 
Theorem~\ref{thm:Miura}, which states that
\bez
 \phi' = \phi - \theta \, \Omega^{-1} \eta + C \, , \qquad
 g' = (I - \theta \,  \Omega^{-1} \Gamma^{-1} \eta) \, g \, ,
\eez
where $C$ is any $\d$-constant $2 \times 2$ matrix, solve the Miura equation, 
and Proposition~\ref{prop:pKN-1}, which says that
\bez
   u' := \phi'_{21} \, , \qquad 
   v' := \imag \frac{g'_{12}}{g'_{22}} 
       = \imag \frac{ (1 - \theta_1 \Omega^{-1} \Gamma^{-1} \eta_1 ) \, \tilde{g}_{12}  - \theta_1 \Omega^{-1} \Gamma^{-1} \eta_2 \, \tilde{g}_{22}}{(1 - \theta_2 \Omega^{-1} \Gamma^{-1} \eta_2) \, \tilde{g}_{22} - \theta_2 \Omega^{-1} \Gamma^{-1} \eta_1 \, \tilde{g}_{12} } \, , 
\eez
solve the system (\ref{pKN-1}). Noting that $C$ is $\d$-constant if and only 
if $C = \mathrm{diag}(c_1,c_2)$ with $c_{ix}=0$, $i=1,2$, and using (\ref{u,v}), 
we arrive at (\ref{u',v'}). 
\end{proof}

\begin{remark}
\label{rem:Delta,Gamma_Jnf}
(\ref{FL_linsys}) is invariant under
\bez
    &&  \theta_i \mapsto \theta_i S \, , \qquad \eta_i \mapsto T \eta_i \qquad i=1,2 \, , \\
    && \Delta \mapsto S^{-1} \Delta S \, , \qquad \Gamma \mapsto T \Gamma T^{-1} \, ,
\eez
with any constant invertible $n \times n$ matrices $S$ and $T$. As a consequence, without restriction of generality, 
we may assume that $\Delta$ and $\Gamma$ are in Jordan normal form. \hfill $\Box$
\end{remark}

\begin{remark}
For trivial seed, i.e., $u=0$ and $v = 0$, (\ref{u',v'}) reduces to (3.24) in \cite{LLZ24}.  \hfill $\Box$
\end{remark}

\begin{example}
The Fokas-Lenells equation (\ref{FLeq}) admits the following plane wave solution,
\be
   u = A \, e^{\imag [ \alpha \, x + (2|A|^2 - \alpha^{-1}) \, t ]} \, ,  
        \label{FL_pw}
\ee
with a complex constant $A$ and a real constant $\alpha \neq 0$. 
Setting $v = u^\ast$, we have a solution of the coupled Fokas-Lenells equations (\ref{pKN-1}). 
To simplify matters, in the following we impose the constraint
\bez
      \alpha = |A|^{-2} \, ,
\eez
so that 
\bez
u = A \, e^{\imag \varphi} \, , \qquad
\varphi := |A|^{-2} \, x + |A|^2 \, t \, .
\eez
Then the linear system (\ref{FL_linsys}) simplifies to
\bez
&&   \theta_{1x} \, \Delta + \frac{1}{2} \theta_1 = 0 \, , \qquad
\theta_{1t} - \frac{1}{2} \theta_1 \Delta = 0 \, , \nonumber \\
&& \theta_{2x} \, \Delta - \frac{1}{2} \theta_2 = - \imag |A|^{-2} u \, \theta_1  \, , 
\qquad
\theta_{2t} + \frac{1}{2} \theta_2 \Delta = - u \, \theta_1 \, , \nonumber \\
&&  \Gamma \, \eta_{2x} + \frac{1}{2} \, \eta_2 = 0 \, , \qquad
\eta_{2t} - \frac{1}{2} \Gamma \, \eta_2 = 0 \, , \nonumber \\
&& \Gamma \, \eta_{1x} - \frac{1}{2} \, \eta_1 = \imag |A|^{-2} \, u \, \eta_2  \, , \qquad
\eta_{1t} + \frac{1}{2} \Gamma \, \eta_1 = u \, \eta_2  \, .
\eez
Hence
\bez
 && \theta_1 = a_1 \, e^{-\Phi(\Delta)} \, , \qquad
     \theta_2 = - A \, e^{\imag \varphi} \, a_1 \, (\Delta + \imag |A|^2 I)^{-1}   
    \, e^{-\Phi(\Delta)} + a_2 \, e^{\Phi(\Delta)} \, , \\
 && \eta_2 = e^{-\Phi(\Gamma)} \, b_2 \, ,  \qquad \;
    \eta_1 = A \, e^{\imag \varphi} \, (\Gamma + \imag |A|^2 I)^{-1} 
    e^{-\Phi(\Gamma)} \, b_2 + e^{\Phi(\Gamma)} \, b_1 \, ,  
\eez
where $I$ is the $n \times n$ identity matrix, $a_1,a_2$ are constant $n$-component row vectors, $b_1,b_2$ constant $n$-component column vectors, 
and
\bez
    \Phi(\Delta) := \frac{1}{2} (\Delta^{-1} x - \Delta t) \, .
\eez
Assuming that $\Delta$ and $\Gamma$ have no eigenvalue in common (``spectrum condition"), the Sylvester 
equation (\ref{Sylv}) has a unique solution $\Omega$. According to 
Theorem~\ref{thm:bDT_KN-1},
\be
  u' = A \, e^{\imag \varphi} - \theta_2 \Omega^{-1} \eta_1 \, , 
\qquad 
 v' = \frac{ A^\ast e^{-\imag \varphi} \, (1 - \theta_1 \Omega^{-1} \Gamma^{-1} \eta_1 )  
	-\imag \theta_1 \Omega^{-1} \Gamma^{-1} \eta_2 }{ 1 - \theta_2 \Omega^{-1} \Gamma^{-1} \eta_2 + \imag A^\ast e^{-\imag \varphi} \, \theta_2 \Omega^{-1} \Gamma^{-1} \eta_1 } 
	\, , \label{cFL_pw_sol}
\ee
yields an infinite set of solutions of the coupled Fokas-Lenells equations (\ref{pKN-1}). Although we started with a seed solution $(u,v)$, 
satisfying $v = u^\ast$, the generated solution $(u',v')$ typically does not
satisfy this condition. In fact, it is a difficult task to find conditions 
to be imposed on the parameters such that $v'=u'^\ast$ holds. We will solve this problem in Section~\ref{sec:red}. 

If $\Delta = \mathrm{diag}(\delta_1,\ldots,\delta_n)$ and 
$\Gamma = \mathrm{diag}(\gamma_1,\ldots,\gamma_n)$, $\gamma_i \neq \delta_j$, then
$\Omega$ is the Cauchy-like matrix with components 
\bez
     \Omega_{ij} = \frac{\eta_{1i} \, \theta_{1j} + \eta_{2i} \, \theta_{2j}}{\gamma_i - \delta_j} \qquad i,j=1,\ldots,n \, .
\eez
More generally, (\ref{cFL_pw_sol}) yields solutions of (\ref{pKN-1}) 
for any pair $(\Delta,\Gamma)$ of constant matrices, satisfying the spectrum condition (so that the Sylvester equation (\ref{Sylv}) has a unique 
solution). Without restriction, $\Delta$ and $\Gamma$ can both be taken in Jordan normal form. \hfill $\Box$
\end{example}

\subsection{Miura transformations from the coupled Fokas-Lenells equations 
	\label{subsec:Miura}
to (\ref{phi_eq}) and (\ref{g_eq})}
Using the bidifferential calculus in Section~\ref{sec:FLsys}, the integrable equation (\ref{phi_eq}) takes the form
\be
&&  \phi_{11tx} = (\phi_{12} \, \phi_{21})_x = - \phi_{22tx} \, ,  \label{phi_eq_1} \\
&&  \phi_{12xt} = \phi_{12} \, \big( 1 + (\phi_{22} - \phi_{11})_x \big) 
\, , \label{phi_eq_2} \\
&&  \phi_{21xt} = \phi_{21} \, \big( 1 + (\phi_{22} - \phi_{11})_x \big) 
\label{phi_eq_3} \, .
\ee
By integrating (\ref{phi_eq_1}) with respect to $t$, the remaining two 
equations can be written as\footnote{Here we should better replace 
the integral by an auxiliary function $w$ and add the equation 
$w_t = (\phi_{22} - \phi_{11})_x$. }  
\be
\phi_{12 xt} = \phi_{12} - 2 \phi_{12} \int (\phi_{12} \, \phi_{21})_x \, dt 
\, , \qquad
\phi_{21 xt} = \phi_{21} - 2 \phi_{21} \int (\phi_{12} \, \phi_{21})_x \, dt 
\, .   \label{AKNS-1}
\ee 

With the identification $u = \phi_{21}$ (cf. (\ref{u,v})), imposing the 
relations (\ref{Miura_3}) and (\ref{Miura_4}), (\ref{phi_eq_3}) becomes the first 
of the coupled Fokas-Lenells equations (\ref{pKN-1}). For $u \neq 0$, 
(\ref{phi_eq_1}) becomes the second of (\ref{pKN-1}). As a consequence of 
these, (\ref{phi_eq_2}) is then identically satisfied. Hence, 
(\ref{Miura_3}) and (\ref{Miura_4}) constitute a Miura transformation from 
the coupled Fokas-Lenells equations (\ref{pKN-1}) to the first member of the negative part of the AKNS hierarchy, which is (\ref{AKNS-1}) \cite{AKNS73}. 
Any solution $(u,v)$ of (\ref{pKN-1}) thus determines a solution of the latter.
\vspace{.2cm}

The integrable equation (\ref{g_eq}), evaluated with the bidifferential 
calculus in Section~\ref{sec:FLsys}, takes the form
\be
&&   (\tilde{g}_{22}\,\tilde{g}_{11t}-\tilde{g}_{21}\,\tilde{g}_{12t} )_x=0 \, ,  \label{g_eq_1} \\
&&  (\tilde{g}_{11}\,\tilde{g}_{22t}-\tilde{g}_{12}\, \tilde{g}_{21t})_x=0
\, , \label{g_eq_2} \\
&& \tilde{g}_{11}\,\tilde{g}_{12} - (\tilde{g}_{11}\,\tilde{g}_{12t}-\tilde{g}_{12}\,\tilde{g}_{11t})_x=0 \, ,
\label{g_eq_3} \\
&& \tilde{g}_{21}\,\tilde{g}_{22}- (\tilde{g}_{22}\,\tilde{g}_{21t}-\tilde{g}_{21}\,\tilde{g}_{22t})_x=0 \, . 
\label{g_eq_4} 
\ee
As a consequence of $\det(\tilde{g}) = 1$, (\ref{g_eq_1}) and (\ref{g_eq_2}) are equivalent. Together with
\bez
      \tilde{g}_{12} = - \imag v \, \tilde{g}_{22} \, ,
\eez
the relations (\ref{Miura_1}) and (\ref{Miura_2}) constitute a Miura transformation from the coupled Fokas-Lenells equations (\ref{pKN-1}) to the system (\ref{g_eq_1})-(\ref{g_eq_4}). Indeed, (\ref{g_eq_1}), (\ref{g_eq_3}) 
and (\ref{g_eq_4}) become, respectively,
\bez
  && \big( v \, ( u_{xt} - u - 2 \imag u u_x v ) \big)_x = 0  \, , \\
  && \imag ( v_{xt} - v + 2 \imag u v v_x ) 
     + \big( v^2 \, ( u_{xt} - u - 2 \imag u v u_x ) \big)_x = 0 \, , \\
  && ( u_{xt} - u - 2 \imag u u_x v )_x = 0 \, . 
\eez
If $(u,v)$ solves the coupled Fokas-Lenells equations (\ref{pKN-1}), 
then the latter equations are satisfied.

\begin{remark}
We recall that, setting
\bez
     \tilde{g} = \left( \begin{array}{rr} \cos(\Theta/2) & - \sin(\Theta/2) \\
     	\sin(\Theta/2) & \cos(\Theta/2) \end{array} \right) \, ,
\eez
with a real (or complex) function $\Theta$, 
the system (\ref{g_eq_1})-(\ref{g_eq_4}) reduces to the (complex) sine-Gordon equation $\Theta_{xt} = \sin(\Theta)$. 
The above reduction for $\tilde{g}$ is compatible with the restrictions 
$\phi_{22} = - \phi_{11}$ and $\phi_{21} = \phi_{12}$. The Miura system 
then reads
\bez
    \phi_{12} = - \frac{1}{2} \Theta_t \, , \qquad
    \phi_{12x} =  - \frac{1}{2} \sin(\Theta) \, , \qquad
    \phi_{11x} = \frac{1}{2} \big( 1-\cos(\Theta) \big) \, .
\eez
By a transformation of variables, (\ref{phi_eq_1})-(\ref{phi_eq_3}) are 
equivalent to the sharp line self-induced transparency (SIT) equations
(see \cite{DKMH11acta}, for example, and references cited there). 
\hfill $\Box$
\end{remark}

\section{Reduction to the Fokas-Lenells equation}
\label{sec:red}
As already mentioned in the introduction, the system (\ref{pKN-1}) reduces to the Fokas-Lenells equation (\ref{FLeq}) via 
\be
          v = u^\ast \, .    \label{red}
\ee
We have to implement this reduction in the vectorial binary Darboux transformation, expressed in Theorem~\ref{thm:bDT_KN-1}, in order 
to obtain a vectorial Darboux transformation for the Fokas-Lenells equation. Because of the strong asymmetry between the generated 
$u'$ and $v'$ in (\ref{u',v'}), this is rather difficult to achieve. 

In Section~\ref{subsec:FLred_trivial_seed}, we will therefore first deal with the case of trivial seed, i.e., $u=v=0$. This also makes contact with 
\cite{Ye+Zhang23,LLZ24}, which do not proceed beyond this case.
In Section~\ref{subsec:FLred_nontrivial_seed} we then turn to the general case.

\subsection{The case of vanishing seed}
\label{subsec:FLred_trivial_seed}
The following proposition is obtained as a special case of Theorem~\ref{thm:bDT_KN-1}.

\begin{proposition}  
\label{prop:FL_bDT_zero_seed}
Let $\Gamma$ be an invertible constant $n \times n$ matrix such that $\Gamma$ and $-\Gamma^\dagger$ (where $^\dagger$ denotes the conjugate transpose) have no 
eigenvalue in common (``spectrum condition"). Let 
\be
         \eta_1 = e^{- \frac{1}{2} (\Gamma^{-1} x + \Gamma t)} \, a_1 \, , \qquad
         \eta_2 = e^{\frac{1}{2} (\Gamma^{-1} x + \Gamma t)} \, a_2 \, ,  \label{FL0_eta}
\ee
with constant $n$-component column vectors $a_1,a_2$.
Furthermore, let $\Omega_1$ and $\Omega_2$ be the (unique) $n \times n$ matrix solutions of the Lyapunov equations
\be
        \Gamma \Omega_1 + \Omega_1 \Gamma^\dagger = \imag \eta_1 \eta_1^\dagger \Gamma^\dagger  \, , \qquad
   \Gamma \Omega_2 + \Omega_2 \Gamma^\dagger = \eta_2 \eta_2^\dagger \, .
               \label{FL0_Lyapunov}
\ee
Then, in any open set of $\mathbb{R}^2$ where $\Omega = \Omega_1 + \Omega_2$ is invertible, 
\be
       u' = - \eta_2^\dagger \, \Omega^{-1} \eta_1   \label{FL_sol}
\ee
solves the Fokas-Lenells equation (\ref{FLeq}).
\end{proposition}
\begin{proof}\footnote{The proof uses some clever observations from the proof of Theorem 2 in \cite{LLZ24}, where the counterpart $K$ of our matrix $\Gamma$ has been 
unnecessarily restricted to be diagonal.}
In Theorem~\ref{thm:bDT_KN-1} we set (cf. \cite{LLZ24})
\be
      \Delta = - \Gamma^\dagger \, , \qquad \theta_1 = \imag \eta_1^\dagger \Gamma^\dagger \, , \qquad
     \theta_2 = \eta_2^\dagger \, .   \label{red_cond_zero_seed}
\ee
The linear system (\ref{FL_linsys}) then reduces to 
\bez
       \Gamma \, \eta_{1x} + \frac{1}{2} \eta_1 = 0 \, , \qquad 
       \eta_{1t} + \frac{1}{2} \Gamma \, \eta_1 = 0 \, , \qquad
      \Gamma \, \eta_{2x} - \frac{1}{2} \eta_2 = 0 \, , \qquad
     \eta_{2t} - \frac{1}{2} \Gamma \, \eta_2 = 0 \, .
\eez
The general solution is given by (\ref{FL0_eta}). 
Using the spectrum condition, (\ref{Sylv}) splits into the two equations (\ref{FL0_Lyapunov}), where $\Omega = \Omega_1 + \Omega_2$. The spectrum 
condition guarantees uniqueness of the solution of the Lyapunov equation, 
hence $\Omega_2^\dagger = \Omega_2$, and, writing
\be
       \Omega_1 = \imag \tilde{\Omega}_1 \Gamma^\dagger \, ,
        \label{tOm1}
\ee
we have $\Gamma \tilde{\Omega}_1 + \tilde{\Omega}_1 \Gamma^\dagger = \eta_1 \eta_1^\dagger$ and thus $\tilde{\Omega}_1^\dagger = \tilde{\Omega}_1$.
It follows that
\bez
    (\Gamma \Omega_1)^\dagger = - \Gamma \Omega_1 \, .  
\eez
We obtain
\bez
   \eta_2 \, \eta_2^\dagger = \Gamma \, \Omega_2 + \Omega_2 \, \Gamma^\dagger 
   = \Gamma \, \Omega + \Omega^\dagger \, \Gamma^\dagger 
     - (\Gamma \Omega_1 + \Omega_1^\dagger \Gamma^\dagger) 
   = \Gamma \, \Omega + \Omega^\dagger \, \Gamma^\dagger \, .
\eez
For trivial seed,  (\ref{u',v'}) reduces to (\ref{FL_sol}) and
\bez
    v' = \frac{\eta_1^\dagger \Gamma^\dagger \Omega^{-1} \Gamma^{-1} \eta_2}{1
         - \eta_2^\dagger \Omega^{-1} \Gamma^{-1} \eta_2} \, .
\eez
Now 
\bez
   u'^\ast \, ( 1 - \eta_2^\dagger \Omega^{-1} \Gamma^{-1} \eta_2 ) 
   = - \eta_1^\dagger  (\Omega^\dagger)^{-1} \eta_2 +  \eta_1^\dagger  (\Omega^\dagger)^{-1} \, \eta_2 \, \eta_2^\dagger \, \Omega^{-1} \Gamma^{-1} \eta_2 
 = \eta_1^\dagger \Gamma^\dagger \Omega^{-1} \Gamma^{-1} \eta_2  
\eez
shows that $u'^\ast = v'$. 
Since the spectrum condition for $\Gamma$ is assumed to hold, 
the linear differential equations (\ref{FL_Om})
are a consequence of the Lyapunov equations. According to Theorem~\ref{thm:bDT_KN-1}, $u'$ solves (\ref{FLeq}). 
\end{proof}

\begin{remark}
Without restriction of generality, we may assume that $\Gamma$ has Jordan normal form, see Remark~\ref{rem:Delta,Gamma_Jnf}. Then $\Gamma$ has the block-diagonal 
structure
\be
      \Gamma = \mbox{block-diag}(\Gamma_1,\ldots,\Gamma_k) \, ,  \label{Gamma_bd}
\ee
with Jordan blocks $\Gamma_j$, $j=1,\ldots,k$, with eigenvalue $\gamma_j$. For a given choice of $\Gamma$, essentially one only has to solve the above Lyapunov equations. Assuming the spectral condition, there is always a unique solution.
\hfill $\Box$
\end{remark}

\begin{example}
Let $\Gamma$ be diagonal, i.e., $\Gamma = \mathrm{diag}(\gamma_1, \ldots, \gamma_n)$, and $\gamma_i^\ast \neq - \gamma_j$ 
for $i,j=1,\ldots,n$. The solutions of the above two Lyapunov equations are then given by the Cauchy-like matrices with entries
\bez
       \Omega_{1ij} = \imag \frac{ \eta_{1i} \eta_{1j}^\ast \gamma_j^\ast}{\gamma_i + \gamma_j^\ast} \, , \qquad
       \Omega_{2ij} =  \frac{ \eta_{2i} \eta_{2j}^\ast}{\gamma_i + \gamma_j^\ast} 
       \, ,
\eez
so that
\bez
     \Omega_{ij} = \frac{ \imag \eta_{1i} \eta_{1j}^\ast \gamma_j^\ast + \eta_{2i} \eta_{2j}^\ast}{\gamma_i + \gamma_j^\ast} 
  =  \frac{ \imag a_{1i}\,  a_{1j}^\ast \gamma_j^\ast \, e^{-\varphi(\gamma_i) - \varphi(\gamma_j^\ast)} + a_{2i} \, a_{2j}^\ast \,  e^{\varphi(\gamma_i)+\varphi(\gamma_j^\ast)} }{\gamma_i + \gamma_j^\ast} \, ,  
\eez
where 
\be
      \varphi(\gamma) = \frac{1}{2} \big( \gamma^{-1} \, x + \gamma \, t \big) 
      \, .   \label{varphi(gamma)}
\ee      
Now 
\bez
     u' = - \sum_{i,j=1}^n a_{2i}^\ast \, a_{1j} \, (\Omega^{-1})_{ij} \, 
          e^{ \varphi(\gamma_i^\ast) - \varphi(\gamma_j) } 
\eez
is an $n$-soliton solution of the Fokas-Lenells equation (also see 
\cite{Lene10,Mats12a,Guo+Ling12,WXL20,LZLX20,LWZ22} for other expressions and derivations). 
For $n=1$, we recover its single soliton solution
\bez
  u' = - \frac{ a_1 a_2^\ast \, (\gamma + \gamma^\ast) }{ 
  	      |a_2|^2 \, e^{2 \varphi(\gamma)} 
     	+ \imag |a_1|^2 \gamma^\ast \, e^{-2 \varphi(\gamma^\ast)} } \, .
\eez
A similar expression can be found in \cite{LWZ22}. Also see 
\cite{Lene+Foka09,Lene09,Mats12a} for earlier appearances. 
\hfill $\Box$
\end{example}

\begin{example}
	\label{ex:2x2Jordan}
Let $\Gamma$ be a $2 \times 2$ Jordan block with eigenvalue $\gamma$, i.e.,
\bez
	\Gamma = \left( \begin{array}{cc} \gamma & 0 \\ 1 & \gamma 
	              \end{array} \right) \, .
\eez 	
Then we have
\bez
 \eta_1 =\left( \begin{array}{c} a_{11}   \\ 
 	 a_{12} + a_{11} \rho  
 	            \end{array} \right) \, e^{ -\varphi(\gamma)} \, , \qquad 
 \eta_2 = \left( \begin{array}{c} a_{21}  \\ 
 	a_{22} - a_{21} \rho  
                \end{array} \right) \, e^{\varphi(\gamma)}\, ,
\eez
where $\varphi(\gamma)$ is again given by (\ref{varphi(gamma)}), and 
\bez
    \rho := \frac{1}{2} \big( \gamma^{-2} x - t \big) \, .
\eez 
Assuming $\mathrm{Re}(\gamma) \neq 0$, the solution of the Lyapunov equation is given by 
\bez
  \Omega = \Omega_1 + \Omega_2 = \imag \tilde{\Omega}_1 \Gamma^\dagger + \Omega_2 
\eez
(cf. (\ref{tOm1})), with (cf. Example~3.1 in \cite{CMH17}) 
\bez
  \tilde{\Omega}_1 = \frac{1}{\kappa}
  \left( \begin{array}{cc} |\eta_{11}|^2 & \eta_{11} \tilde{\eta}_{12}^\ast  \\
  \eta_{11}^\ast \tilde{\eta}_{12} & |\tilde{\eta}_{12}|^2 + \kappa^{-2} |\eta_{11}|^2 \end{array} \right)  \, , \qquad
  \Omega_2 = \frac{1}{\kappa}
   \left( \begin{array}{cc} |\eta_{21}|^2 & \eta_{21} \tilde{\eta}_{22}^\ast  \\
   	\eta_{21}^\ast \tilde{\eta}_{22} & |\tilde{\eta}_{22}|^2 + \kappa^{-2} |\eta_{21}|^2 \end{array} \right) \, ,
\eez 
where $\eta_{ij}$ is the $j$th component of the vector $\eta_i$, and 
\bez
    \tilde{\eta}_{i2} := \eta_{i2} - \kappa^{-1} \eta_{i1} \, , 
    \qquad 
    \kappa := 2 \, \mathrm{Re}(\gamma) \, .
\eez
Hence, $\Omega$ has the components
\bez
 \Omega_{11} &=& \kappa^{-1} \big( \imag \gamma^\ast \, |\eta_{11}|^2 
      + |\eta_{21}|^2 \big) \, , \\
 \Omega_{12} &=& \kappa^{-1} \big( \imag |\eta_{11}|^2 + \imag \gamma^\ast \eta_{11} \tilde{\eta}_{12}^\ast + \eta_{21} \tilde{\eta}_{22}^\ast \big)
        \, , \\
 \Omega_{21} &=& \kappa^{-1} \big( \imag \gamma^\ast \eta_{11}^\ast \tilde{\eta}_{12} + \eta_{21}^\ast \tilde{\eta}_{22} \big) \, , \\
 \Omega_{22} &=& \kappa^{-1} \big( \imag \eta_{11}^\ast \tilde{\eta}_{12} + \imag \gamma^\ast (|\tilde{\eta}_{12}|^2 + \kappa^{-2} |\eta_{11}|^2)
 + |\tilde{\eta}_{22}|^2 + \kappa^{-2} |\eta_{21}|^2 \big) \, .       
\eez
Using
\bez
   \Omega^{-1} = \frac{1}{\mathrm{det}(\Omega)}
 \left( \begin{array}{cc} \Omega_{22} & -\Omega_{12} \\ -\Omega_{21} & \Omega_{11} \end{array} \right)\, ,
\eez
(\ref{FL_sol}) leads to the following solution of the Fokas-Lenells equation, 
\bez
   u' = \frac{F}{\mathrm{det}(\Omega)} 
\eez
with
\bez
 F &=& e^{-2 \imag \mathrm{Im}(\varphi(\gamma))} \, \Big( (a_{11} a_{21}^\ast |\rho|^2-a_{11}a_{22}^\ast\rho + a_{12} a_{21}^\ast \rho^\ast - a_{12} a_{22}^\ast) \, \Omega_{11} + a_{21}^\ast (a_{11} \rho + a_{12}) \, \Omega_{12} \\
&& -a_{11} (a_{21}^\ast \rho^\ast - a_{22}^\ast) \, \Omega_{21} - a_{11} a_{21}^\ast \, \Omega_{22} \Big) \, , \\
 \mathrm{det}(\Omega)&=& \kappa^{-4}\,\big(|a_{21}|^4e^{4\,\mathrm{Re}(\varphi(\gamma))}-\gamma^{\ast 2}|a_{11}|^4e^{-4\,\mathrm{Re}(\varphi(\gamma))}\big)+\imag\kappa^{-2}\big(\gamma^\ast |2a_{11}a_{21}\rho-\mathrm{det}(a_1,a_2)|^2\\
 &&+2|a_{11}|^2|a_{21}|^2\rho\big)+\imag\kappa^{-2}a_{11}^\ast a_{21}^\ast\big(2\gamma^\ast\kappa^{-2}a_{11}a_{21}-\mathrm{det}(a_1,a_2)\big)\, ,
\eez
where $\mathrm{det}(a_1,a_2) = a_{11} a_{22} - a_{12} a_{21}$.
Besides the exponential dependence, this also depends on the independent variables $x$ and $t$ through the linear expression $\rho$. An equivalent 
solution has been obtained in \cite{LWZ22} via Hirota's bilinearization method (``double-pole solution"), also see 
\cite{ZQH23}. Fig.~\ref{fig:2x2Jordan} shows a plot of the absolute value of $u'$ for special parameters. 
\begin{figure}[h]
\begin{center}
	\includegraphics[scale=.5]{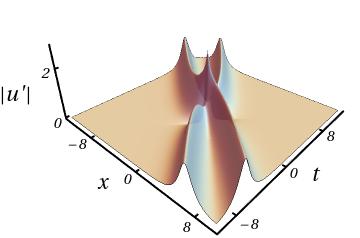} 
	\parbox{15cm}{
	\caption{Plot of the absolute value of a ``positon" solution of the Fokas-Lenells equation 
	from the class in Example~\ref{ex:2x2Jordan}. Here we chose $\gamma=a_{11}=a_{12}=a_{21}=a_{22}=1$. \label{fig:2x2Jordan} } 
	             } 
\end{center}
\end{figure}\hfill $\Box$
\end{example}

\begin{remark}
More generally, let $\Gamma$ be an $n \times n$ Jordan block with eigenvalue $\gamma$, i.e.,
\be
 \Gamma_{(n)} = \left( \begin{array}{ccccc} \gamma & 0&\ldots&\ldots&0 \\ 1 & \gamma&\ddots&\ddots&0\\
	0&\ddots&\ddots&\ddots&\vdots\\
	\vdots&\ddots&\ddots&\ddots&0\\
	0&\ldots&0&1&\gamma \end{array} \right) \, .   \label{Jordan_n}
\ee 		
For increasing size Jordan blocks, the solutions of the Lyapunov equation 
are nested. This means that the solution has the structure
\bez
     \Omega_{(n)} = \left( \begin{array}{cc} \Omega_{(n-1)} & \alpha_n \\ \beta_n & \omega_n 
     \end{array} \right) \, ,
\eez	
with certain $(n-1)$-component (column, respectively row) vectors $\alpha_n, \beta_n$, and a scalar $\omega_n$, all involving the solution 
of the linear system. Using the solution $\Omega_{(2)}$ in Example~\ref{ex:2x2Jordan}, one can easily compute $\Omega_{(3)}$, and so forth.
We refer to \cite{CMH17} for details. \hfill $\Box$
\end{remark}

Solutions, which we obtain directly using an $n \times n$ Jordan block as ``spectral matrix", are obtained in other approaches by taking special limits of ordinary $n$-soliton solutions, in which the eigenvalues coincide. Sometimes resulting solutions are called ``positons".

\subsection{Vectorial Darboux transformation for the Fokas-Lenells equation}
\label{subsec:FLred_nontrivial_seed}

\begin{theorem}
	\label{thm:FL_DT}
Let $\Gamma$ be an invertible constant $n \times n$ matrix, such that $\Gamma$ and $-\Gamma^\dagger$ have no eigenvalue in common. Let $u$ be a solution of the Fokas-Lenells equation (\ref{FLeq}). 
Furthermore, let $\eta_1$ and $\tilde{\eta}_2$ be $n$-component column vector solutions of the linear system
\be
	&&  \Gamma \, \eta_{1x} + \frac{1}{2} \eta_1 - u_x \, \tilde{\eta}_2 = 0 \, , \qquad\quad
	\eta_{1t} - \big( \imag |u|^2 I - \frac{1}{2} \Gamma \big) \, \eta_1 - u \, \tilde{\eta}_2 = 0
	 \, ,   \nonumber \\
   && \Gamma \tilde{\eta}_{2x} - \frac{1}{2} \tilde{\eta}_2 + \imag u^\ast_x \, \Gamma \eta_1 = 0 \, , \qquad
     \tilde{\eta}_{2t} + \big( \imag |u|^2 I - \frac{1}{2} \Gamma \big) \, \tilde{\eta}_2 - \imag u^\ast \Gamma \eta_1 = 0 \, .
	\label{FLeq_linsys}
\ee
Furthermore, let $\Omega$ be an $n \times n$ matrix solution of the Lyapunov equation 
\be
   \Gamma \Omega + \Omega \Gamma^\dagger 
 = \imag \eta_1 \eta_1^\dagger \Gamma^\dagger 
   + \tilde{\eta}_2 \, \tilde{\eta}_2^\dagger \, . \label{FLeq_Lyap}
\ee
Then, in any open set of $\mathbb{R}^2$ where $\Omega$ is invertible,
\be
    u' = u - \tilde{\eta}_2^\dagger \, \Omega^{-1} \eta_1   \label{FLeq_u'}
\ee
solves the Fokas-Lenells equation (\ref{FLeq}).
\end{theorem}
\begin{proof}
By extending the reduction conditions (\ref{red_cond_zero_seed}) to
non-vanishing seed $u$ via
\be
\Delta = - \Gamma^\dagger \, , \quad
\theta_1 = \eta_1^\dagger \, (\imag \Gamma^\dagger + |u|^2 I) - \imag u^\ast \, \eta_2^\dagger \, , \quad
\theta_2 = \imag u \, \eta_1^\dagger + \eta_2^\dagger \, ,  \label{FL_red}
\ee
the four equations of (\ref{FL_linsys}), involving $\theta_1$ and $\theta_2$,  by setting $v = u^\ast$ become
\bez
&& (\Gamma + \imag |u|^2 I) \Big( \Gamma \, \eta_{1x} + (\frac{1}{2} + \imag u_x u^\ast) \, \eta_1 - u_x \, \eta_2 \Big) \\
&& \hspace{1cm} - u \, \Big( \Gamma \, \eta_{2x} - (\frac{1}{2} + \imag u_x u^\ast) \, \eta_2 
+ \imag u^\ast \, (1 + \imag u_x u^\ast) \, \eta_1 \Big) = 0 \, , \\
&& \Big( \Gamma \, \eta_{2x} - (\frac{1}{2} + \imag u_x u^\ast) \, \eta_2 
+ \imag u^\ast \, (1 + \imag u_x u^\ast) \, \eta_1 \Big)
- \imag u^\ast \, \Big(\Gamma \, \eta_{1x} + (\frac{1}{2} + \imag u_x u^\ast) \, \eta_1 - u_x \, \eta_2 \Big) = 0 \, , \\
&& (\Gamma + \imag |u|^2 I) \Big( \eta_{1t} + \frac{1}{2} \Gamma \, \eta_1 - u \, \eta_2 \Big) 
- u \, \Big( \eta_{2t} - \frac{1}{2} \Gamma \, \eta_2 - (\imag v_t - u v^2) \, \eta_1 \Big) = 0 \, , \\
&&  \Big( \eta_{2t} - \frac{1}{2} \Gamma \, \eta_2 - (\imag v_t - u v^2) \, \eta_1 \Big) - \imag u^\ast \, \Big( \eta_{1t} + \frac{1}{2} \Gamma \, \eta_1 - u \, \eta_2 \Big) = 0 \, .
\eez
They are satisfied as a consequence of the equations for $\eta_1$ and 
$\eta_2$ in (\ref{FL_linsys}). The linear system in Theorem~\ref{thm:bDT_KN-1} 
thus reduces to 
\bez
  &&  \Gamma \, \eta_{1x} + (\frac{1}{2} + \imag u_x u^\ast) \, \eta_1 - u_x \, \eta_2 = 0 \, , \qquad
  \eta_{1t} + \frac{1}{2} \Gamma \, \eta_1 - u \, \eta_2 = 0
  \, ,  \\
  && \Gamma \, \eta_{2x} - (\frac{1}{2} + \imag u_x u^\ast) \, (\eta_2 - \imag u^\ast \eta_1) 
  + \frac{1}{2} \imag u^\ast \, \eta_1 = 0 \, , \qquad
  \eta_{2t} - \frac{1}{2} \Gamma \, \eta_2 + (\imag u_t + |u|^2 u)^\ast \, \eta_1 = 0 \, .
\eez
In terms of 
\bez
     \tilde{\eta}_2 := \eta_2 - \imag u^\ast \eta_1  \, ,
\eez
the latter system can be equivalently replaced by (\ref{FLeq_linsys}).
The Sylvester equation (\ref{Sylv}) becomes the Lyapunov equation (\ref{FLeq_Lyap}). We have 
$\Omega = \Omega_1 + \Omega_2$, where 
\bez
 \Gamma \Omega_1 + \Omega_1 \Gamma^\dagger = \imag \eta_1 \eta_1^\dagger \Gamma^\dagger \, , \qquad
 \Gamma \Omega_2 + \Omega_2 \Gamma^\dagger = \tilde{\eta}_2 \, \tilde{\eta}_2^\dagger \, .
\eez
As in the proof of Proposition~\ref{prop:FL_bDT_zero_seed}, we obtain
\be
    \tilde{\eta}_2 \, \tilde{\eta}_2^\dagger = \Gamma \Omega + \Omega^\dagger \Gamma^\dagger \, .  \label{key}
\ee
$u'$ in (\ref{u',v'}) takes the form (\ref{FLeq_u'}), and $v'$ can be written as
\bez
     v' = \frac{G}{H} \, , 
\eez
with
\bez
 G := u^\ast + \big( \Gamma \eta_1 - u \, \tilde{\eta}_2 \big)^\dagger \, (\Gamma \Omega)^{-1} \tilde{\eta}_2  \, , \qquad
 H := 1 - \tilde{\eta}_2^\dagger \, (\Gamma \Omega)^{-1} \tilde{\eta}_2 \, .
\eez
Next we compute
\bez
   u'^\ast H &=& \big( u^\ast - \eta_1^\dagger \, \Omega^{\dagger -1} \tilde{\eta}_2 \big) \big( 1 - \tilde{\eta}_2^\dagger \, (\Gamma \Omega)^{-1} \tilde{\eta}_2 \big) \\
   &=& u^\ast - u^\ast \tilde{\eta}_2^\dagger \, (\Gamma \Omega)^{-1} \tilde{\eta}_2
     - \eta_1^\dagger \, \Omega^{\dagger -1} \tilde{\eta}_2
     + \eta_1^\dagger \, \Omega^{\dagger -1} \tilde{\eta}_2 \tilde{\eta}_2^\dagger 
     \, (\Gamma \Omega)^{-1} \tilde{\eta}_2 = G
   \, ,
\eez 
where we used (\ref{key}) in the last step. We have thus shown that $v' = u'^\ast$. 
Since, according to Theorem~\ref{thm:bDT_KN-1}, $(u',v')$ solves the 
coupled Fokas-Lenells equations (\ref{pKN-1}), it follows that 
$u'$ solves the Fokas-Lenells equation (\ref{FLeq}). 
\end{proof}

\begin{remark}
Dropping the spectrum condition, Theorem~\ref{thm:FL_DT} remains true if we supplement the assumptions there by the equations	
\be
  \Omega_x  = \imag \eta_{1x} \, \eta_1^\dagger 
    + (\tilde{\eta}_{2x} + \imag u^\ast_x \eta_1 ) \, \tilde{\eta}_2^\dagger \,  \Gamma^{\dagger -1} \, , \qquad 
  \Omega_t = \frac{1}{2} \Big( -\imag \eta_1 \eta_1^\dagger  \Gamma^\dagger + 2 \imag u^\ast \eta_1 \tilde{\eta}_2^\dagger
  + \tilde{\eta}_2 \tilde{\eta}_2^\dagger \Big) \, ,
    \label{FLeq_Om_deriv}
\ee	
and require in addition that (\ref{key}) holds. 

(\ref{FLeq_Om_deriv}) results from (\ref{FL_Om}), using (\ref{FL_red}). As a consequence of (\ref{FLeq_Om_deriv}) and the linear system (\ref{FLeq_linsys}), it follows that
\bez
   \big( \Gamma \Omega + \Omega^\dagger \Gamma^\dagger - \tilde{\eta}_2 \, \tilde{\eta}_2^\dagger \big)_x = 0 \, , \qquad
   \big( \Gamma \Omega + \Omega^\dagger \Gamma^\dagger - \tilde{\eta}_2 \, \tilde{\eta}_2^\dagger \big)_t = 0 \, ,
\eez
so that (\ref{FLeq_Om_deriv}) is indeed compatible with (\ref{key}).

If the spectrum condition for $\Gamma$ is satisfied, 
then (\ref{key}) and (\ref{FLeq_Om_deriv}) are automatically 
satisfied as a consequence of the Lyapunov equation (\ref{FLeq_Lyap}).
\hfill $\Box$
\end{remark}

\begin{remark}
	\label{rem:FL_superposition}
Let $(\Gamma^{(i)},\eta^{(i)}_1,\tilde{\eta}^{(i)}_2,\Omega^{(i)})$,
$i=1,2$, be data that determine solutions ${u'}^{(i)}$ of the 
Fokas-Lenells equation for the same seed solution $u$. Let
\bez
   \Gamma = \left( \begin{array}{cc} \Gamma^{(1)} & 0 \\
   0 & \Gamma^{(2)} \end{array} \right) \, , \quad
   \eta_1 = \left(\begin{array}{c} \eta^{(1)}_1 \\ \eta^{(2)}_1
   	\end{array} \right) \, , \quad
   \tilde{\eta}_2 = \left(\begin{array}{c} \tilde{\eta}^{(1)}_2 \\ \tilde{\eta}^{(2)}_2
   \end{array} \right) \, , \quad	
   \Omega = \left( \begin{array}{cc} \Omega^{(1)} & \Omega^{(12)} \\
   	\Omega^{(21)} & \Omega^{(2)} \end{array} \right) \, .
\eez
The linear system (\ref{FLeq_linsys}) is then satisfied.	
In order that the above composed data also determine a 
solution, which is then a nonlinear  
superposition of ${u'}^{(1)}$ and ${u'}^{(2)}$, essentially 
one has therefore only to determine the matrices $\Omega^{(12)}$ and $\Omega^{(21)}$. The Lyapunov equation (\ref{FLeq_Lyap}) 
reduces to the two Sylvester equations
\bez
 &&  \Gamma^{(1)} \Omega^{(12)} + \Omega^{(12)} \Gamma^{(2)\dagger}
   = \imag \eta^{(1)}_1 \eta^{(2)\dagger}_1 \Gamma^{(2)\dagger}
   + \tilde{\eta}^{(1)}_2 \tilde{\eta}^{(2)\dagger}_2 \, , 
    \\
 && \Gamma^{(2)} \Omega^{(21)} + \Omega^{(21)} \Gamma^{(1)\dagger}
 = \imag \eta^{(2)}_1 \eta^{(1)\dagger}_1 \Gamma^{(1)\dagger}
 + \tilde{\eta}^{(2)}_2 \tilde{\eta}^{(1)\dagger}_2 \, .
\eez
If $\Gamma^{(1)}$ and $-\Gamma^{(2)\dagger}$ have no eigenvalue in common, these equations possess unique solutions. Otherwise additional equations, arising from (\ref{key}) and (\ref{FLeq_Om_deriv}), have to be taken into account.	
\hfill $\Box$
\end{remark}

\begin{remark}
For $n=1$, the linear system (\ref{FLeq_linsys}) can be written as
\bez
    \left(\begin{array}{c} \eta_1 \\ \tilde{\eta}_2 \end{array} \right)_x
  = \gamma^{-1} \left(\begin{array}{cc} - \frac{1}{2} & u_x \\
     - \imag \gamma \, u_x^\ast & \frac{1}{2} \end{array} \right) 
    \left(\begin{array}{c} \eta_1 \\ \tilde{\eta}_2 \end{array} \right) \, , 
    \quad
 \left(\begin{array}{c} \eta_1 \\ \tilde{\eta}_2 \end{array} \right)_t
= \left(\begin{array}{cc} \imag |u|^2 - \frac{1}{2} \gamma & u \\
	 \imag \gamma \, u^\ast & - (\imag |u|^2 - \frac{1}{2} \gamma) \end{array} \right) 
\left(\begin{array}{c} \eta_1 \\ \tilde{\eta}_2 \end{array} \right) \, .   
\eez
This is a Lax pair for the Fokas-Lenells equation with spectral parameter 
$\gamma$, which in (\ref{FLeq_linsys}) is promoted to a matrix. \hfill $\Box$
\end{remark}

\section{Solutions of the Fokas-Lenells equation from a plane wave seed}
\label{sec:FLeq_pw_seed}
As the seed in the Darboux transformation, expressed in Theorem~\ref{thm:FL_DT}, we choose the plane wave solution (\ref{FL_pw}), i.e., 
\bez
      u = A \, e^{\imag \varphi} \, , \qquad 
      \varphi := \alpha \, x + (2 |A|^2 - \alpha^{-1}) \, t \, ,
\eez	
where $\alpha \in \mathbb{R} \setminus \{0\}$ and 
$A \in \mathbb{C} \setminus \{0\}$. Writing
\bez
   \eta_1 = e^{\imag \varphi/2} \, \chi_1 \, , \qquad
   \tilde{\eta}_2 = e^{-\imag \varphi/2} \, \chi_2 \, , 
\eez
with column vectors $\chi_1, \chi_2$, (\ref{FLeq_linsys}) 
decouples to
\be
 && \chi_{1xx} + R^2 \, \chi_1 = 0 \, ,  \qquad
\chi_{1t} = - \imag \alpha^{-1} \Gamma \chi_{1x} \, , 
          \label{FLeq_ex_pw_2}  \\
 &&  \chi_2 = - \frac{\imag}{\alpha A} \Big( \Gamma \, \big(\chi_{1x} + \frac{1}{2} \imag \alpha \chi_1 \big) + \frac{1}{2} \chi_1 \Big) 
     \, ,  \label{FLeq_ex_pw_1} 
\ee
where
\bez
    R^2 := \frac{1}{4} (\alpha^2 I - \Gamma^{-2}) 
    + \imag \alpha^2 \big( |A|^2 - \frac{1}{2} \alpha^{-1} \big) \Gamma^{-1} \, .
\eez
If $R^2$ is invertible, it possesses an invertible matrix square root $R$.
The two equations in (\ref{FLeq_ex_pw_2}) are then solved by  
\be
  \chi_1 = e^{\imag \Phi} \, a_1 + e^{-\imag \Phi} \, b_1 \, , 
  \qquad \Phi := R \, x  - \imag \alpha^{-1} \Gamma R \, t  \, ,
       \label{FLeq_pw_chi1}
\ee
with constant $n$-component column vectors $a_1,b_1$, and 
(\ref{FLeq_ex_pw_1}) leads to
\be
     \chi_2 = - \frac{\imag}{\alpha A} \, \Big( 
   \big( \frac{1}{2} I + \imag \Gamma \, (\frac{1}{2} \alpha I + R) \big) \, e^{\imag \Phi} \, a_1
   + \big( \frac{1}{2} I + \imag \Gamma \, (\frac{1}{2} \alpha I - R) \big) \, e^{-\imag \Phi} \, b_1 \Big) \, .
      \label{FLeq_pw_chi2}
\ee     
It remains to obtain $\Omega$ from the Lyapunov equation (\ref{FLeq_Lyap}), 
which is now
\bez
   \Gamma \Omega + \Omega \Gamma^\dagger 
   = \imag \chi_1 \chi_1^\dagger \Gamma^\dagger 
   + \chi_2 \, \chi_2^\dagger \, .
\eez
Then, via (\ref{FLeq_u'}), we get
\be
 u' = e^{\imag \varphi} \big( A - \chi_2^\dagger \Omega^{-1} \chi_1 \big)
    = e^{\imag \varphi} \Big( A - \sum_{i,j=1}^n \chi_{2i}^\ast \, 
    (\Omega^{-1})_{ij} \, \chi_{1j} \Big) 
    \, .          \label{FL_u'_pw}
\ee
\vspace{.2cm}

\noindent
We have to distinguish the following two cases: \\
\textbf{(1)} 
If the spectrum condition for $\Gamma$ holds, then the Lyapunov equation possesses a unique solution, irrespective of the 
$n \times n$ matrix on its right hand side. In particular, 
if $\Gamma = \mathrm{diag}(\gamma_1,\ldots,\gamma_n)$, with 
$\gamma_i \neq -\gamma_j^\ast$, $i,j=1,\ldots,n$, 
the solution $\Omega$ is given by the Cauchy-like matrix with components 
\bez
    \Omega_{ij} = \frac{\imag \chi_{1i} \, \chi_{1j}^\ast \, \gamma_j^\ast 
    	+ \chi_{2i} \, \chi_{2j}^\ast}{\gamma_i + \gamma_j^\ast} 
    	\, .
\eez
\textbf{(2)} 
If the spectrum condition does \emph{not} hold, the condition that the Lyapunov equation has a solution imposes constraints on the 
matrix on its right hand side. Moreover, its solution is then not 
unique. This freedom is fixed by (\ref{FLeq_Om_deriv}), which 
now takes the form
\be
 && \Omega_x = (\imag \chi_{1x} - \frac{1}{2} \alpha \chi_1) \chi_1^\dagger
   + (\chi_{2x} - \frac{1}{2} \imag \alpha \chi_2 + \alpha A^\ast \chi_1)
     \chi_2^\dagger \Gamma^{\dagger -1} \, , \nonumber \\
 && \Omega_t = \frac{1}{2} \big( -\imag \chi_1 \chi_1^\dagger \Gamma^\dagger + 2 \imag A^\ast \chi_1 \chi_2^\ast + \chi_2 \chi_2^\dagger \big) \, ,    \label{FLeq_pw_Om_deriv}  
\ee
and (\ref{key}), which is
\be
    \Gamma \Omega + \Omega^\dagger \Gamma^\dagger 
    = \chi_2 \chi_2^\dagger \, .   \label{key_pw}
\ee
This degenerate case of the Lyapunov equation is important in order 
to recover dark solitons (also see \cite{Manas96} for the case of the NLS equation).

\subsection{Single breather and dark soliton solutions}
	\label{subsec:FLeq_1sol}
For $n=1$, we write $\Gamma = \gamma$ and $R=r$. 
Assuming $r \neq 0$, the solution of the linear 
system, given by (\ref{FLeq_pw_chi1}) and (\ref{FLeq_pw_chi2}),  
becomes
\bez
  && \chi_1 = e^{\imag \Phi} a_1 + e^{-\imag \Phi} b_1 \, , 
      \\
  && \chi_2 = - \frac{\imag}{\alpha A} \, \Big( 
  \big( \frac{1}{2} + \imag \gamma \, (\frac{1}{2} \alpha + r) \big) \, e^{\imag \Phi} \, a_1
  + \big( \frac{1}{2} + \imag \gamma \, (\frac{1}{2} \alpha - r) \big) \, e^{-\imag \Phi} \, b_1 \Big) \, ,
\eez
where 
\bez
  \Phi = r \, (x - \imag \gamma  \, \alpha^{-1} t) \, , \qquad
  r = \pm \sqrt{ \frac{1}{4} (\alpha^2 - \gamma^{-2}) 
	+ \imag \gamma^{-1} \alpha^2 \, ( |A|^2 - \frac{1}{2} \alpha^{-1} )  } 
    \, . 
\eez

\paragraph{(1)}
Let $\kappa := 2 \mathrm{Re}(\gamma) \neq 0$. Then 
$\Omega = (\imag \gamma^\ast |\chi_1|^2 + |\chi_2|^2)/\kappa$
is everywhere invertible and the generated solution of the Fokas-Lenells equation, 
\bez
  u' =  e^{\imag\varphi} \, \Big( A - 
   \frac{\kappa \, \chi_1 \, \chi_2^\ast}{\imag \gamma^\ast |\chi_1|^2 + |\chi_2|^2} \Big) \, , 
\eez
is thus regular on $\mathbb{R}^2$.
As	special cases, it includes Akhmediev- and Kuznetsov-Ma-type breathers.

\paragraph{(2)}
Let $\mathrm{Re}(\gamma)=0$. We write $\gamma = -\imag k$ with real $k$. The Lyapunov equation then becomes the constraint
\bez
	  |\chi_2|^2 = k \, |\chi_1|^2 \, ,
\eez
which requires $k>0$. Since it has to hold for all $x$ and $t$, it can only be satisfied if either $a_1=0$ or $b_1=0$. 
Let us choose $b_1=0$, so that
\bez
 \chi_1 = e^{\imag \Phi} a_1 \, , \qquad 
 \chi_2 = -\frac{\imag}{\alpha A} \big( \frac{1}{2} + k \, (\frac{1}{2} \alpha + r) \big) \, e^{\imag\Phi} a_1 \, .
\eez
Then the above constraint amounts to
\bez
 \frac{1}{\alpha^2 |A|^2} \left| \frac{1}{2} + k \, \big( \frac{1}{2} \alpha + r \big) \right|^2 - k = 0 \, ,
\eez
which is automatically satisfied if 
\bez
  r = \pm \frac{1}{2} \sqrt{ k^{-2}  
   	- 4 k^{-1} \alpha^2 \, ( |A|^2 - \frac{1}{2} \alpha^{-1} )  
   	+ \alpha^2 }
\eez
is imaginary. This condition requires
\bez
  \big( k - 2 ( |A|^2 - \frac{1}{2} \alpha^{-1} ) \big)^2 
      < 4 \, |A|^2 \, ( |A|^2 - \alpha^{-1} )  \, ,
\eez 
so that, in particular, 
\bez
      |A|^2 > \alpha^{-1} \, .
\eez
Noting that now $\Phi = r \, (x-k t/\alpha)$ and $\Phi^\ast = - \Phi$, since $r^\ast = -r$, from (\ref{FLeq_pw_Om_deriv}) we obtain
\bez
 \Omega_x = \frac{|a_1|^2}{2k} \big( 1 - k \, (\alpha+2r) \big) \, e^{2\imag \Phi} \, ,\qquad 
 \Omega_t = - \frac{|a_1|^2}{2\alpha} \big(1 - k \, (\alpha+2r)\big) \, e^{2\imag\Phi} \, .
\eez
This integrates to 
\bez
	\Omega = - \frac{\imag |a_1|^2 \big(1-k\,(\alpha+2r)\big)}{4rk} e^{2\imag \Phi} + c \, |a_1|^2 \, ,
\eez
with a constant $c$, which (\ref{key_pw}) requires to be real. The generated solution of the Fokas-Lenells equation is
\bez
  u' &=& e^{\imag \varphi} \Big(A + \frac{1}{\alpha A^\ast} \, \frac{2 r k \, \big( 1+k\,(\alpha-2r)\big) \, e^{2\imag \Phi}}{\big( 1-k\,(\alpha+2r) \big) \, e^{2\imag \Phi} + 4 \imag r \, k \, c} \Big)  \\ 
  &=& e^{\imag \varphi} \Big( A + \frac{r k \, \big(1+k \, (\alpha-2r) \big)}{ \alpha A^\ast \big(1-k\,(\alpha+2r)\big) } 
  \big( 1 + \mathrm{tanh}(\imag\Phi + \delta) \big) \Big)
   \, ,
\eez
with a complex constant $\delta$ given by
\bez
  e^{2\delta} = \frac{ 1-k\,(\alpha+2r)}{4 \, k \, (\imag r) \, c} \, ,
\eez
assuming $c \neq 0$. 
The solution is regular on $\mathbb{R}^2$. 
The absolute square of $u'$ is given by
\bez
 |u'|^2 = |A|^2 + \frac{2 k \, (\imag r) \, e^{2\imag \Phi}}{ c \, \alpha \, |e^{2(\imag \Phi+\delta)} + 1|^2 } 
\eez
(we recall that $\imag \Phi$ is real).
Thus $u'$ represents a dark soliton ($|u'|^2$ shows a sink in the constant background density) if 
$c \, \alpha \, (\imag r) < 0$,  and a bright soliton if 
$c \, \alpha \, (\imag r) > 0$. 
Fig. \ref{fig:dark soliton} shows the modulus of a dark soliton solution from this class.	
	
\begin{figure}[h]
\begin{center}
\includegraphics[scale=.4]{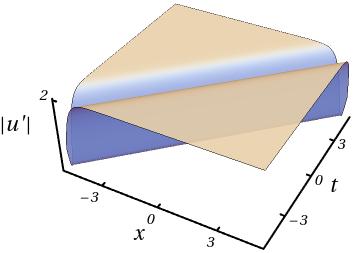} 
	\parbox{15cm}{
	\caption{Plot of the absolute value of a dark soliton solution of the Fokas-Lenells equation from the class in Section~\ref{subsec:FLeq_1sol}(2). Here we chose $\alpha=a_1=c=1,b_1=0,A=2,k=1,$ so that $r=\sqrt{3}\,\imag$. 	
		\label{fig:dark soliton} } 
		} 
\end{center}
\end{figure}

\subsection{Multiple dark/bright solitons}
\label{subsec:dark-bright_solitons}
Let $|A|^2 > \alpha^{-1}$ and 
$\Gamma = - \imag \mathrm{diag}(k_1,\ldots,k_n)$ with
$k_j >0$, $j=1,\ldots,n$, and $k_i \neq k_j$, $i,j=1,\ldots,n$. Let
\bez
&& \chi^{(j)}_1 = e^{\imag \Phi_j} a_j \, , \qquad 
\chi^{(j)}_2 = -\frac{\imag}{\alpha A} \big( \frac{1}{2} + k_j \, (\frac{1}{2} \alpha + r_j) \big) \, e^{\imag \Phi_j} a_j 
\, , \qquad
\Phi_j = r_j \, (x - k_j \alpha^{-1} t)  \, , \\
&& \Omega_j = - \frac{\imag |a_j|^2 \big(1 - k_j \, (\alpha + 2 r_j) \big)}{4 r_j k_j} e^{2\imag \Phi_j} + c_j \, |a_j|^2 \, ,
\eez
where $a_j \in \mathbb{C} \setminus \{0\}$, $c_j \in \mathbb{R} \setminus \{0\}$, and for each $j$, $k_j$ is 
such that
\bez
r_j = \pm \frac{1}{2} \sqrt{ k_j^{-2}  
	- 4 k_j^{-1} \alpha^2 \, ( |A|^2 - \frac{1}{2} \alpha^{-1} )  
	+ \alpha^2 } 
\eez
is imaginary. Setting
\bez
 \chi_1 = \left( \begin{array}{c} \chi^{(1)}_1 \\ \vdots \\ \chi^{(n)}_1
\end{array} \right) \, , \quad
\chi_2 = \left( \begin{array}{c} \chi^{(1)}_2 \\ \vdots \\ \chi^{(n)}_2
\end{array} \right) \, ,  
\eez
and
\bez
 && \Omega_{jj} = \Omega_j \qquad  j=1,\ldots,n \, , \\ 
 && \Omega_{ij} = \frac{\imag}{k_i-k_j} 
 \big( - \chi^{(i)}_1 \chi^{(j)\ast}_1 k_j + \chi^{(i)}_2 \chi^{(j)\ast}_2  \big)  \qquad i \neq j \, ,
\eez 
it follows that 
\bez
   u' &=& e^{\imag \varphi} \big( A 
   - \chi_2^\dagger \, \Omega^{-1} \chi_1 \big)  
\eez
is an $n$-soliton solution of the Fokas-Lenells equation. 
Depending on the signs of $\alpha c_1 (\imag r_1), \ldots,
\alpha c_n (\imag r_n)$, it represents a superposition of $n$ 
dark or bright solitons.  

For $n=2$, we have
\bez
u' = e^{\imag \varphi} \Big( A - \frac{1}{\det(\Omega)} 
 \left[
\chi^{(1)\ast}_2 \big( \Omega_2 \chi^{(1)}_1 - \Omega_{12} \chi^{(2)}_1 \big) + \chi^{(2)\ast}_2 \big( \Omega_1 \chi^{(2)}_1 - \Omega_{21} \chi^{(1)}_1 \big)  \right] \Big)   \, .     
\eez
If $\alpha c_j (\imag r_j) < 0$ (respectively $\alpha c_j (\imag r_j) > 0$), $j=1,2$, this represents a pair of dark (respectively bright) solitons.  
We have a superposition of a dark and a bright 
soliton if $c_1 c_2 r_1 r_2 >0$.  
Also see Fig.~\ref{fig:dark-bright_solitons}.

\begin{figure}[h]
\begin{center}
\includegraphics[scale=.5]{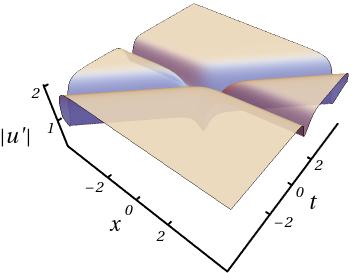} 
\hspace{1cm}
\includegraphics[scale=.5]{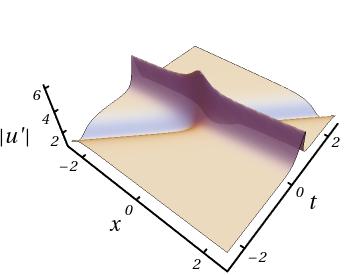}
\parbox{15cm}{
\caption{Plot of the absolute value of a dark-dark (left plot) and a dark-bright (right plot) soliton pair solution of the Fokas-Lenells equation from the class in Section~\ref{subsec:dark-bright_solitons}. Here we chose 
$\alpha=a_1=a_2=c_1 = k_1 =1$, $k_2=10$, $A=2$ and 
$c_2 =1$, respectively $c_2=-1$.
\label{fig:dark-bright_solitons} } 
		} 
\end{center}
\end{figure}

Dark solitons of the Fokas-Lenells equation 
have been obtained before in different ways \cite{Veks11,Mats12b,ZFH13,DTSN23}.

\subsection{Rogue waves}

Another important class of solutions are rogue waves. We will show that they appear when the matrix $R^2$ is degenerate. The system (\ref{FLeq_ex_pw_2}) and (\ref{FLeq_ex_pw_1}) then has solutions 
\emph{not} covered by the expressions for $\chi_1$ and $\chi_2$ in 
(\ref{FLeq_pw_chi1}) and (\ref{FLeq_pw_chi2}). 

\begin{example}
	\label{ex:rw}
Let $n=1$ and $R^2 = 0$, so that
\be
	\gamma = \imag (\alpha^{-1} - 2 |A|^2)
	\pm 2 |A| \sqrt{ \alpha^{-1} - |A|^2 } \, .   \label{gamma_rw}
\ee
In this degenerate case, the solution of the linear system (\ref{FLeq_ex_pw_2})-(\ref{FLeq_ex_pw_1}) is found to be
\bez
	\chi_1 = a \, \rho + b \, , \qquad 
	\chi_2 = - \frac{\imag}{\alpha A} \big( \gamma \, a 
	+ \frac{1}{2} (1+\imag \alpha \gamma) \, (a \, \rho + b) \big) \, , 
\eez
with complex constants $a$, $b$, and 
\be
     \rho := x - \imag \gamma \, \alpha^{-1} t \, .  \label{rw_rho}
\ee
The corresponding solution of the Lyapunov equation is 
\bez
  \Omega = \frac{1}{\kappa} \, \Big( 
  \imag \gamma^\ast \, \left| a \, \rho + b \right|^2 
  + \frac{1}{\alpha^2 |A|^2} \, \Big| \gamma \, a 
  + \frac{1}{2} (1+\imag \alpha \gamma) \, (a \, \rho + b) \Big|^2 \Big) \, ,
\eez
with $\kappa = 2 \mathrm{Re}(\gamma) \neq 0$. The last condition requires that 
the discriminant in (\ref{gamma_rw}) is positive, i.e.,
\bez
     |A|^2 < \alpha^{-1}  \, .
\eez
Since $\Omega \neq 0$ on $\mathbb{R}^2$, we obtain a regular quasi-rational solution
\bez
    u' =  e^{\imag\varphi} \, \Big( A - \frac{1}{\Omega} \, \chi_1 \, 
      \chi_2^\ast \Big) \, .
\eez
This is a counterpart of the Peregrine breather solution of the focusing NLS equation. It models a rogue wave, which has the 
characteristic feature of being localized in space and time. 
Also see Fig.~\ref{fig:rw}. 

\begin{figure}[h]
	\begin{center}
	\includegraphics[scale=.5]{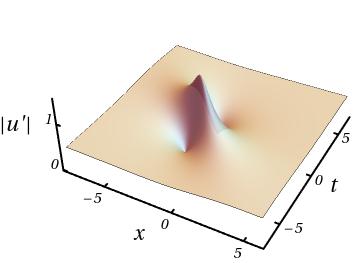} 
	\parbox{15cm}{
	\caption{Plot of the absolute value of a solution of the Fokas-Lenells equation from the class in Example~\ref{ex:rw}. Here we chose $\alpha=a=b=1,A=\frac{1}{2}$, 
	so that $\gamma=\frac{\sqrt{3}}{2}+\frac{1}{2}\,\imag$. 
	The deformation of the constant density background is localized in both dimensions.      \label{fig:rw} } 
			} 
	\end{center}
\end{figure}
\hfill $\Box$
\end{example}

\begin{example}
	\label{ex:rw2}
Let $n=2$ and 	
\bez
    \Gamma = \left( \begin{array}{cc} \gamma & 0 \\ 1 & \gamma 
\end{array} \right) \, , \quad
   \chi_i = \left( \begin{array}{c} \chi_{i1} \\ \chi_{i2} 
   \end{array} \right) \qquad i=1,2 \, ,
\eez 
where $\gamma$ is given by (\ref{gamma_rw}). The spectrum 
condition requires $\mathrm{Re}(\gamma) \neq 0$, hence
we need $\alpha^{-1} > |A|^2$. Then we have
\bez
     R^2 = \left( \begin{array}{cc} 0 & 0 \\ B & 0 
     \end{array} \right) \, , \qquad
     B := \frac{1}{2 \gamma^3} \big( 1 + \imag \alpha \gamma \, (1-2 \alpha |A|^2) \big)  \, .
\eez
(\ref{FLeq_ex_pw_2}) now takes the form
\bez
    \chi_{11xx} = 0 \, , \quad
    \chi_{11t} = - \imag \alpha^{-1} \gamma \, \chi_{11x} \, \quad
    \chi_{12xx} = - B \, \chi_{11} \, , \quad
    \chi_{12t} = - \imag \alpha^{-1} (\chi_{11x} + \gamma \chi_{12x})
    \, ,
\eez
so that
\bez
  \chi_{11} = a_1 \, \rho + a_0 \, , \qquad
  \chi_{12} = -B \, \big( \frac{1}{6} a_1 \rho^3 +  \frac{1}{2} a_0 \rho^2 + b_1 \rho + b_0  \big) - \frac{a_1}{\gamma} x \, ,  
\eez
with $\rho$ given by (\ref{rw_rho}). Using (\ref{FLeq_ex_pw_1}), 
we obtain
\bez
  \chi_{21} = - \frac{\imag}{\alpha A} \Big( 
    \frac{1}{2} (1 + \imag  \alpha \gamma) \chi_{11} 
    + \gamma a_1 \Big) \, , \qquad
  \chi_{22} = - \frac{\imag}{\alpha A} \Big( a_1 + \frac{1}{2} \imag \alpha \, \chi_{11} + \gamma \chi_{12x} + \frac{1}{2} (1+\imag \alpha \gamma) \chi_{12} \Big) 
    \, ,     
\eez
with complex constants $a_0,a_1,b_0,b_1$. 
The solution of the Lyapunov equation is
\bez
    \Omega = \left( \begin{array}{cc} \kappa^{-1} M_{11} &
   - \kappa^{-2} M_{11} + \kappa^{-1} M_{12} \\[4pt] 
   - \kappa^{-2} M_{11} + \kappa^{-1} M_{21} &
   2 \kappa^{-3} M_{11} - \kappa^{-2} (M_{12}+M_{21})	
     + \kappa^{-1} M_{22} \end{array} \right) \, ,
\eez
where $\kappa = 2 \, \mathrm{Re}(\gamma)$ and
\bez
    M = (M_{ij}) = \imag \chi_1 \chi_1^\dagger \Gamma^\dagger 
    + \chi_2 \, \chi_2^\dagger 
    = \left( \begin{array}{cc} \imag \gamma^\ast |\chi_{11}|^2  
    	+ |\chi_{21}|^2 &
      \imag (|\chi_{11}|^2 + \gamma^\ast \chi_{11}  \chi_{12}^\ast ) + \chi_{21} \chi_{22}^\ast \\[2pt] 
      \imag \gamma^\ast \chi_{12} \chi_{11}^\ast + \chi_{22} \chi_{21}^\ast	 &
      \imag (\chi_{12} \chi_{11}^\ast 
      + \gamma^\ast |\chi_{12}|^2) + |\chi_{22}|^2
    	 \end{array} \right) \, .
\eez
Hence
\bez
    \Omega^{-1} = \frac{1}{\det(\Omega)} 
    \left( \begin{array}{cc} 2 \kappa^{-3} M_{11} - \kappa^{-2} (M_{12}+M_{21}) + \kappa^{-1} M_{22} &
    \kappa^{-2} M_{11} - \kappa^{-1} M_{12} \\[4pt] 
    \kappa^{-2} M_{11} - \kappa^{-1} M_{21} &
    \kappa^{-1} M_{11} \end{array} \right) \, ,
\eez
with
\bez
    \det(\Omega) = \kappa^{-4} M_{11}^2 + \kappa^{-2} \det(M) \, .
\eez
Inserting the expressions for $\chi_1,\chi_2$ and $\Omega^{-1}$ 
in (\ref{FL_u'_pw}), determines a quasi-rational solution of the Fokas-Lenells equation. 
For special data, the absolute value of the solution 
is plotted in Fig.~\ref{fig:rw2}.

\begin{figure}[h]
	\begin{center}
	\includegraphics[scale=.5]{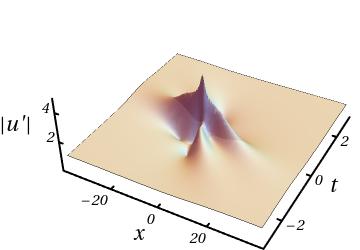} 
	\parbox{15cm}{
	\caption{Plot of the absolute value of a solution of the Fokas-Lenells equation from the class in Example~\ref{ex:rw2}. Here we chose $\alpha=1/4$ and $A=a_0=a_1=b_0=b_1=1$.      \label{fig:rw2} } 
		} 
	\end{center}
\end{figure}
\hfill $\Box$
\end{example}

\begin{remark}
Generalizing the preceding examples, an $n$-th order rogue wave is 
obtained by choosing (\ref{Jordan_n}) for $\Gamma$, with the eigenvalue
$\gamma$ given by (\ref{gamma_rw}). The matrix $R^2$ is then nilpotent 
of order $n$ and one can use results in \cite{CMH17} to elaborate
the corresponding class of solutions of the Fokas-Lenells equation. 
More generally, by taking any $\Gamma$ of the form (\ref{Gamma_bd}), where 
now all Jordan blocks have the same eigenvalue (\ref{gamma_rw}), nonlinear 
superpositions of rogue waves are obtained. Finally, we can also ``superpose"
in this way rogue waves and breathers by allowing also eigenvalues 
of Jordan blocks different from (\ref{gamma_rw}). This quickly 
results in very lengthy expressions, but examples are easily 
worked out with the help of computer algebra. 
\hfill $\Box$
\end{remark}

Rogue waves (of an equivalent form) of the Fokas-Lenells equation  
have also been obtained in \cite{HXP12,XHCP15,WHQGM19}, using a scalar Darboux 
transformation.

\section{Final remarks}
\label{sec:concl}

We derived a vectorial binary Darboux transformation for the  
coupled Fokas-Lenells equations, which is the first member of the ``negative" 
part of the Kaup-Newell hierarchy. Furthermore, we were able to reduce 
it to a vectorial Darboux transformation for the  
Fokas-Lenells equation (\ref{FLeq}). In particular, this greatly improves 
recent publications \cite{Ye+Zhang23,LLZ24}, where a related approach does
not extend beyond the case of vanishing seed. 
\vspace{.2cm}

The bidifferential calculus, which we used in this work, is primarily suited to treat the first member of the negative part of the AKNS hierarchy (cf. \cite{MH23}). The coupled Fokas-Lenells equations are not so straightly embedded in this framework (see (\ref{FL_red})), which is the source of the difficulty we met to restrict the map $(u,v) \mapsto (u',v')$, given by the binary Darboux transformation with fixed parameters, in such a way that it preserves the reduction condition (\ref{red}). Finally, this problem was solved by imposing the rather complicated conditions (\ref{u,v}) on the data entering the binary Darboux transformation.   
Perhaps there is another bidifferential calculus (maybe corresponding to another reduction of the sdYM equations as that used in \cite{LLZ24}), which leads more directly to the coupled Fokas-Lenells equations and then allows a simpler implementation of the reduction (\ref{red}).
\vspace{.2cm}

In Section~\ref{sec:FLeq_pw_seed} we derived solutions of the Fokas-Lenells equation by using a plane wave solution as the seed in the vectorial Darboux 
transformation. This reaches breathers, dark and bright solitons, and rogue waves. As in the NLS case \cite{CMH17}, an $n$-th order rogue wave is characterized by data, where the matrix $\Gamma$ is an $n \times n$ Jordan block with a special eigenvalue (see (\ref{gamma_rw})).     
\vspace{.2cm}

Comprehensive results about solitons on a plane wave background have already been obtained before, in particular by use of Hirota's direct (or bilinearization) method \cite{Mats12b} and, via a Riemann-Hilbert problem, in \cite{Zhang+Tian23}. In the present work, for the first 
time, according to the best of our knowledge, we have reached all 
relevant types of soliton solutions by using a Darboux transformation method.
\vspace{.2cm}

Multi-component generalizations of the coupled Fokas-Lenells 
equations have been considered in \cite{Guo+Ling12,Gerd+Ivan21}. We expect that our approach can be extended to also treat such cases. 
\vspace{.2cm} 

The Fokas-Lenells example, treated in this work, motivates to explore the Miura transformation (\ref{Miura}) also in case of other bidifferential calculi in a similar way.

\end{document}